\documentclass[a4paper,USenglish,cleveref,autoref,thm-restate]{lipics-v2021}
\usepackage{mathtools}
\usepackage{cite}
\usepackage{algorithm}
\usepackage{algpseudocode}
\newcommand{\ang}[1]{\langle #1\rangle}

\renewcommand{\ang}[1]{\langle #1\rangle}

\newcommand{\RE}{\mathbb{R}}            
\newcommand{\eps}{\varepsilon}          
\newcommand{\ST}{\,:\,}                 
\newcommand{\bd}{\partial\kern+1pt}     
\newcommand{\etal}{\textit{et al.}}
\newcommand{\SP}{\kern+1pt}             
\newcommand{\PERP}{\kern-2pt \perp \kern-2pt} 
\newcommand{\TT}{\mathcal{T}}           
\DeclareMathOperator{\interior}{int}    
\DeclareMathOperator{\Vor}{Vor}         
\DeclareMathOperator{\DT}{DT}           

\title{Delaunay Triangulations in the Hilbert Metric}
\titlerunning{Delaunay Triangulations in the Hilbert Metric}

\author{Auguste H. Gezalyan}{Department of Computer Science, University of Maryland, College Park, USA \and \url{~}}{octavo@umd.edu}{https://orcid.org/0000-0002-5704-312X}{}

\author{Soo H. Kim}{Wellesley College, Wellesley, MA  \and \url{~}}{sk111@wellesley.edu}{}{}

\author{Carlos Lopez}{Montgomery Blair High School, Silver Spring, MD  \and \url{~}}{pcmr.carlos.lopez@gmail.com}{}{}

\author{Daniel Skora}{Indiana University Boolmington, Bloomington, IN \and \url{~}}{danskora@iu.edu}{}{}

\author{Zofia Stefankovic}{Stony Brook University, Stony Brook, NY \and \url{~}}{zofia.stefankovic@stonybrook.edu}{}{}

\author{David M. Mount}{Department of Computer Science, University of Maryland, College Park, USA \and \url{https://www.cs.umd.edu/~mount/}}{mount@umd.edu}{https://orcid.org/0000-0002-3290-8932}{}

\authorrunning{Gezalyan, Kim, Lopez, Skora, Stefankovic, and Mount}

\Copyright{Auguste H. Gezalyan, Soo H. Kim, Carlos Lopez, Daniel Skora, Zofia Stefankovic, and David M. Mount, }


\date{\today}

\hideLIPIcs  
\nolinenumbers 

\begin{document}

\maketitle

\begin{abstract}
The Hilbert metric is a distance function defined for points lying within the interior of a convex body. It arises in the analysis and processing of convex bodies, machine learning, and quantum information theory. In this paper, we show how to adapt the Euclidean Delaunay triangulation to the Hilbert geometry defined by a convex polygon in the plane. We analyze the geometric properties of the Hilbert Delaunay triangulation, which has some notable differences with respect to the Euclidean case, including the fact that the triangulation does not necessarily cover the convex hull of the point set. We also introduce the notion of a Hilbert ball at infinity, which is a Hilbert metric ball centered on the boundary of the convex polygon. We present a simple randomized incremental algorithm that computes the Hilbert Delaunay triangulation for a set of $n$ points in the Hilbert geometry defined by a convex $m$-gon. The algorithm runs in $O(n (\log n + \log^3 m))$ expected time. In addition we introduce the notion of the Hilbert hull of a set of points, which we define to be the region covered by their Hilbert Delaunay triangulation. We present an algorithm for computing the Hilbert hull in time $O(n h \log^2 m)$, where $h$ is the number of points  on the hull's boundary. 
\end{abstract}

\section{Introduction}

David Hilbert introduced the Hilbert metric in 1895~\cite{hilbert1895linie}. Given any convex body $\Omega$ in $d$-dimensional space, the Hilbert metric defines a distance between any pair of points in the interior of $\Omega$ (see Section~\ref{sec:prelim} for definitions). The Hilbert metric has a number of useful properties. It is invariant under projective transformations, and straight lines are geodesics. When $\Omega$ is a Euclidean ball, it realizes the Cayley-Klein model of hyperbolic geometry. When $\Omega$ is a simplex, it provides a natural metric over discrete probability distributions (see, Nielsen and Sun \cite{nielsen2019clustering, nielsen2022nonlinear}). An excellent resource on Hilbert geometries is the handbook of Hilbert geometry by Papadopoulos and Troyanov~\cite{papadopoulos2014handbook}.

The Hilbert geometry provides new insights into classical questions from convexity theory. Efficient approximation of convex bodies has a wide range of applications, including approximate nearest neighbor searching both in Euclidean space~\cite{arya2018membership} and more general metrics~\cite{abdelkader2019noneuclid}, optimal construction of $\eps$-kernels~\cite{arya2017kernel}, solving the closest vector problem approxi\-mately~\cite{eisenbrand2011covering, rothvoss2022cvp, eisenbrand2021cvp, naszodi2019covering}, and computing approximating polytopes with low combinatorial complexity~\cite{arya2017complexity, arya2022complexity}. These works all share one thing in common---they approximate a convex body by covering it with elements that behave much like metric balls. These covering elements go under various names: Macbeath regions, Macbeath ellipsoids, Dikin ellipsoids, and $(2,\eps)$-covers. Vernicos and Walsh showed that these shapes are, up to a constant scaling factor, equivalent to Hilbert balls \cite{vernicos2014hilbert, abdelkader2018delone}. In addition, the Hilbert metric behaves nicely in the context flag approximability of convex polytopes as studied by Vernicos and Walsh \cite{vernicos2018flag}.

Other applications of the Hilbert metric include machine learning \cite{nielsen2019clustering}, quantum information theory \cite{reeb2011hilbertquantum}, real analysis \cite{lemmens2013birkhoff}, and optimal mass transport \cite{chen2016entropic}. Despite its obvious appeals, only recently has there been any work on developing classical computational geometry algorithms that operate in the Hilbert metric. Nielsen and Shao characterized balls in the Hilbert metric defined by a convex polygon with $m$ sides~\cite{nielsen2017balls}. Hilbert balls are convex polygons bounded by $2 m$ sides. Nielsen and Shao showed that Hilbert balls can be computed in $O(m)$ time, and they developed dynamic software for generating them. Gezalyan and Mount presented an $O(m n \log n)$ time algorithm for computing the Voronoi diagram of $n$ point sites in the Hilbert polygonal metric~\cite{gezalyan2021voronoi} (see Figure~\ref{fig:voronoi-delaun}(a) and~(b)). They showed that the diagram has worst-case combinatorial complexity of $O(m n)$. Bumpus {\etal} \cite{bumpus2023software} further analyzed the properties of balls in the Hilbert metric and presented software for computing Hilbert Voronoi diagrams.

\begin{figure}[htbp]
    \centerline{\includegraphics[scale=0.40]{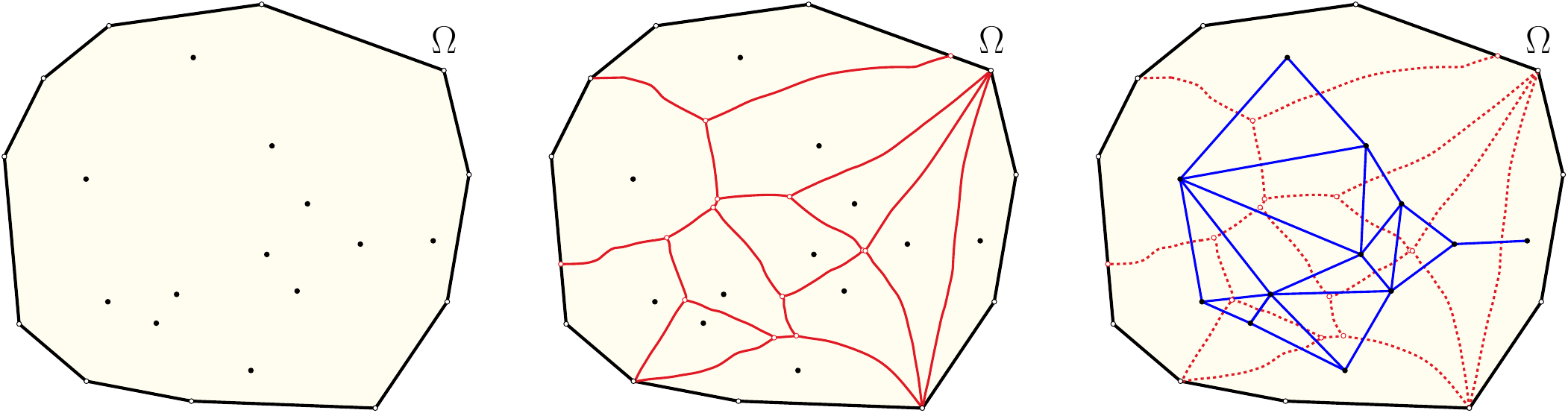}}
    \caption{(a) A convex polygon $\Omega$ and sites, (b) the Voronoi diagram, and (c) the Delaunay triangulation.} 
    \label{fig:voronoi-delaun}
\end{figure} 

In this paper, we present an algorithm for computing the Delaunay triangulation of a set $P$ of $n$ point sites in the Hilbert geometry defined by an $m$-sided convex polygon $\Omega$. The Hilbert Delaunay triangulation is defined in the standard manner as the dual of the Hilbert Voronoi diagram (see Figure~\ref{fig:voronoi-delaun}(c)). Our algorithm is randomized and runs in $O(n (\log n + \log^3 m))$ expected time. Excluding the polylogarithmic term in $m$, this matches the time of the well-known randomized algorithm for Euclidean Delaunay triangulations~\cite{guibas1992randomized}. It is significantly more efficient than the time to compute the Hilbert Voronoi diagram, which is possible because of the smaller output size. A central element of our algorithm is an $O(\log^3 m)$ time algorithm for computing Hilbert circumcircles. 

Unlike the Euclidean case, the Delaunay triangulation does not necessarily triangulate the convex hull of $P$. We also provide a characterization of when three points admit a Hilbert ball whose boundary contains all three points, and introduce the notion of the Hilbert hull, which we define to be the region covered by the Hilbert Delaunay triangulation. We give an algorithm for the Hilbert hull that runs in time $O(n h \log^2 m )$, where $h$ is the number of points on the Hilbert hull's boundary. 

\section{Preliminaries} \label{sec:prelim}

\subsection{The Hilbert Metric and Hilbert Balls}

The Hilbert metric is defined over the interior of a convex body $\Omega$ in $\RE^d$ (that is, a closed, bounded, full dimensional convex set). Let $\bd \Omega$ denote $\Omega$'s boundary. Unless otherwise stated, we assume throughout that $\Omega$ is an $m$-sided convex polygon in $\RE^2$. Given two points $p$ and $q$ in $\Omega$, let $\overline{p q}$ denote the \emph{chord} of $\Omega$ defined by the line through these points.

\begin{figure}[htbp]
    \centerline{\includegraphics[scale=0.40]{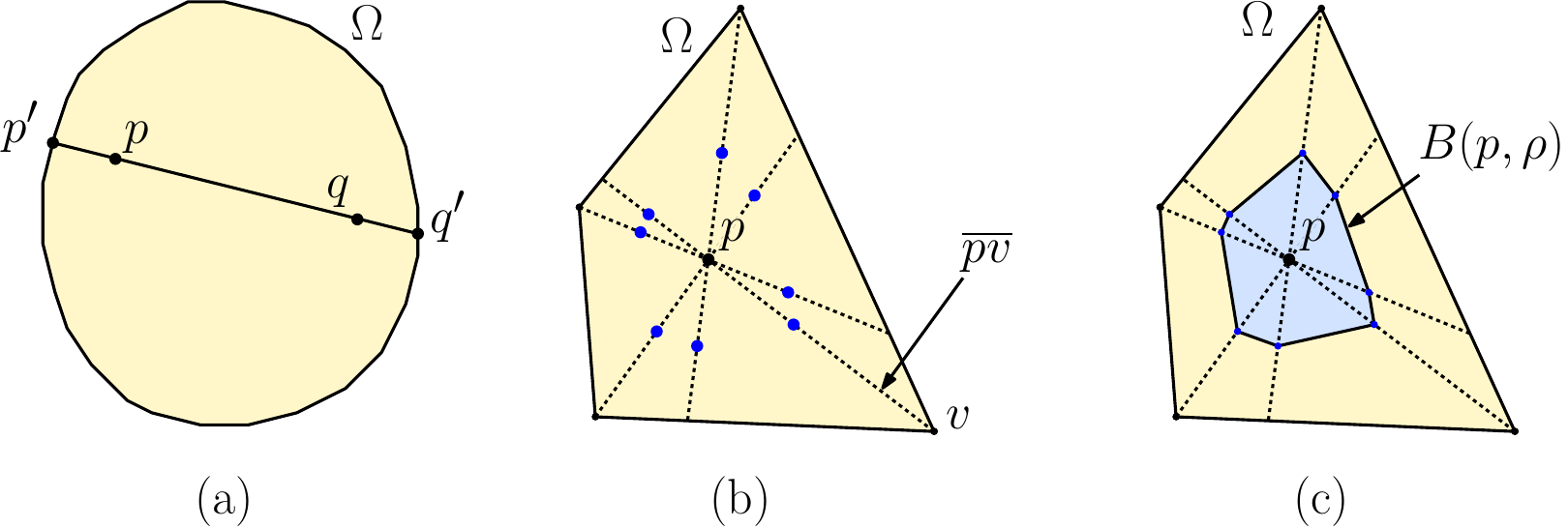}}
    \caption{(a) The Hilbert distance $d_\Omega(p,q)$, (b) spokes defined by $p$, and (c) the Hilbert ball $B(p,\rho)$.} \label{fig:HilbertDef}
\end{figure} 

\begin{definition}[Hilbert metric] 
Given a convex body $\Omega$ in $\RE^d$ and two distinct points $p, q \in \interior(\Omega)$, let $p'$ and $q'$ denote endpoints of the chord $\overline{p q}$, so that the points are in the order $\ang{p', p, q, q'}$. The \emph{Hilbert distance} between $p$ and $q$ is defined
\[
    d_{\Omega}(p,q)
        ~ = ~ \frac{1}{2} \ln \left( \frac{\|q - p'\|}{\|p - p'\|} \frac{\|p - q'\|}{\|q - q'\|} \right),
\]
and define $d_{\Omega}(p,p) = 0$ (see Figure~\ref{fig:HilbertDef}(a)).
\end{definition}

Note that the quantity in the logarithm is the cross ratio $(q, p; p', q')$. Since cross ratios are preserved by projective transformations, it follows that the Hilbert distances are invariant under projective transformations. Straight lines are geodesics, but not all geodesics are straight lines. The Hilbert distance satisfies all the axioms of a metric, and in particular it is symmetric and the triangle inequality holds~\cite{papadopoulos2014handbook}. When $\Omega$ is a probability simplex~\cite{boyd2004convex}, this symmetry distinguishes it from other common methods of calculating distances between probability distributions, such as the Kullback-Leibler divergence~\cite{kullback1951information}. 

Given $p \in \interior(\Omega)$ and $\rho > 0$, let $B(p,\rho)$ denote the Hilbert ball of radius $\rho$ centered at $p$. Nielsen and Shao showed how to compute Hilbert balls~\cite{nielsen2017balls}. Consider the set of $m$ chords $\overline{p v}$ for each vertex $v$ of $\Omega$ (see Figure~\ref{fig:HilbertDef}(b)). These are called the \emph{spokes} defined by $p$. For each spoke $\overline{p v}$, consider the two points at Hilbert distance $\rho$ on either side of $p$. $B(p,\rho)$ is the (convex) polygon defined by these $2 m$ points (see Figure~\ref{fig:HilbertDef}(c)). Given any line that intersects $\Omega$, a simple binary search makes it possible to determine the two edges of $\Omega$'s boundary intersected by the line. Applying this to the line passing through $p$ and $q$, it follows that Hilbert distances can be computed in $O(\log m)$ time.

\subsection{The Hilbert Voronoi Diagram} \label{sec:hilbert-vor}

Given a set $P$ of $n$ point sites in $\interior(\Omega)$, the \emph{Hilbert Voronoi diagram} of $P$ is defined in the standard manner as a subdivision of $\interior(\Omega)$ into regions, called \emph{Voronoi cells}, based on which site of $P$ is closest in the Hilbert distance. It was shown by Gezalyan and Mount~\cite{gezalyan2021voronoi} that each Voronoi cell is star-shaped with respect to its site. Given two sites $p,q \in P$, the \emph{$(p,q)$-bisector} is the set of points that are equidistant from both sites in the Hilbert metric (see Figure \ref{fig:HilbertSec}(a)).

\begin{figure}[htbp]
    \centerline{\includegraphics[scale=0.40]{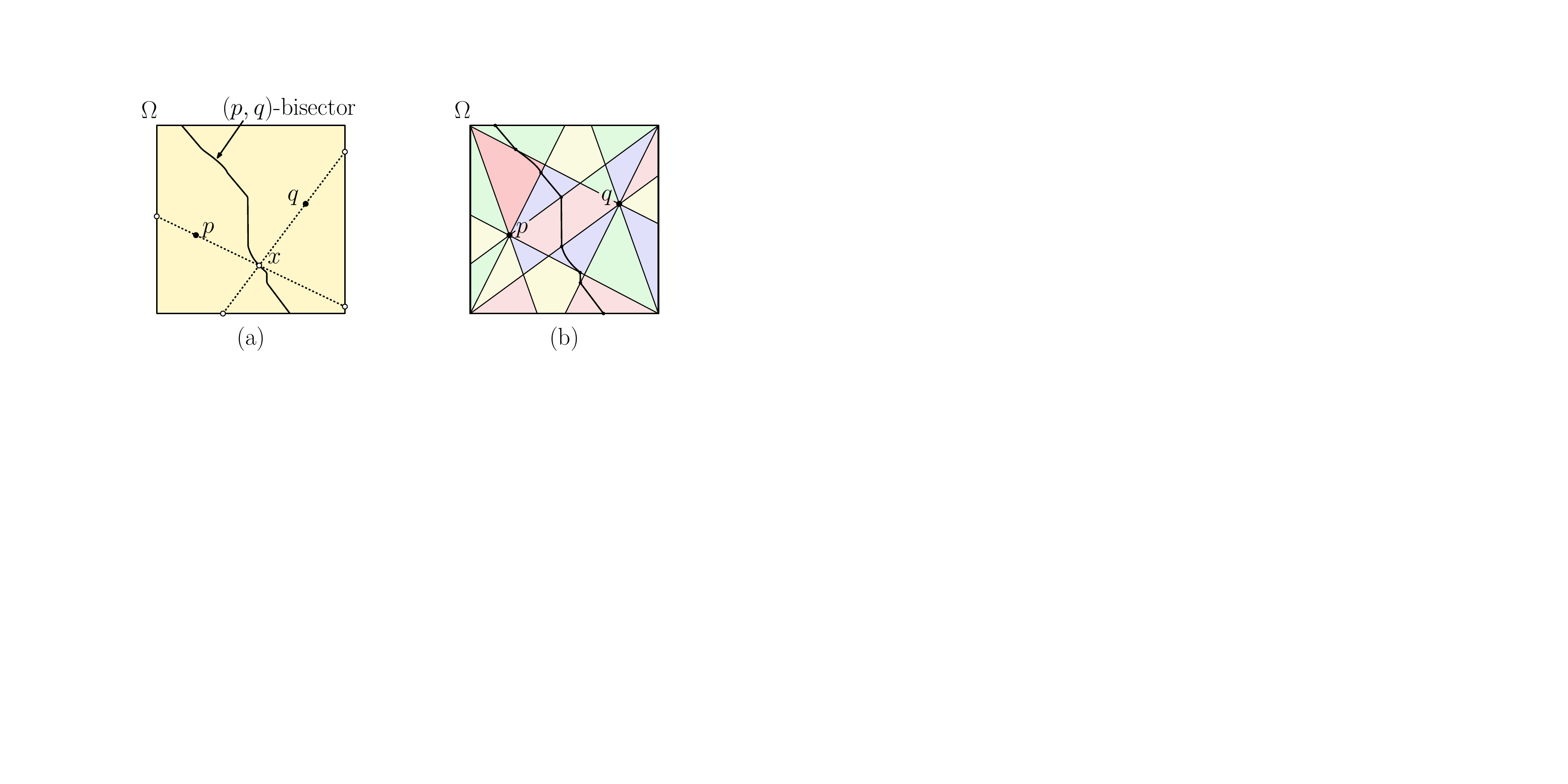}}
    \caption{The $(p,q)$-bisector is a piecewise conic with joints on the spokes of $p$ and $q$.} 
    \label{fig:HilbertSec}
\end{figure}

The distances from any point $x$ on the $(p,q)$-bisector to $p$ and $q$ are determined by the edges of $\bd \Omega$ incident to the chords $\overline{p x}$ and $\overline{q x}$. We can subdivide the interior $\Omega$ into equivalence classes based upon the identity of these edges. It is easy to see that as $x$ crosses a spoke of either $p$ or $q$, it moves into a different equivalence class. It has been shown that within each equivalence class, the bisector is a conic~\cite{gezalyan2021voronoi, bumpus2023software}. It follows that the $(p,q)$-bisector is a piecewise conic, whose joints lie on sector boundaries (see Figure \ref{fig:HilbertSec}(b)). It is also known that the worst-case combinatorial complexity of the Hilbert Voronoi diagram is $\Theta(m n)$~\cite{gezalyan2021voronoi}. 

\subsection{The Hilbert Delaunay Triangulation} \label{sec:hilbert-delaun}

The \emph{Hilbert Delaunay triangulation} of $P$, denoted $\DT(P)$, is defined as the dual of the Hilbert Voronoi diagram, where two sites $p$ and $q$ are connected by an edge if their Voronoi cells are adjacent. Throughout, we make the general-position assumption that no four sites of $P$ are Hilbert equidistant from a single point in $\interior(\Omega)$. It is easy to see that familiar concepts from Euclidean Delaunay triangulations apply, but using Hilbert balls rather than Euclidean balls. For example, two sites are adjacent in the triangulation if and only if there exists a Hilbert ball whose boundary contains both sites, and whose interior contains no sites. (The center of this ball lies on the Voronoi edge between the sites.) Also, three points define a triangle in $DT(P)$ if and only if there is a Hilbert ball whose boundary contains all three points, but is otherwise empty. 

Consider a triangle $\triangle p q r$ in $\Omega$'s interior. A Hilbert ball whose boundary passes through all three points is said to \emph{circumscribe} this triangle. As we shall see in Section~\ref{sec:hilbert-circum}, some triangles admit no circumscribing Hilbert ball, but the following lemma shows that if it does exist, then it is unique. We refer to the boundary of such a ball as a \emph{Hilbert circumcircle}.

\begin{restatable}{lemma}{circumUnique} \label{lem:circum-unique}
There is at most one Hilbert circumcircle for any triangle in $\Omega$'s interior.
\end{restatable}

\begin{proof}
Let $p$, $q$, and $r$ be the vertices of the triangle, which we assume are not collinear. Suppose for the sake of a contradiction there did exist two circumscribing Hilbert balls $B_1$ and $B_2$. Let $c_1$ and $c_2$ denote their respective centers, and let $E = \Omega \setminus (B_1 \cup B_2)$ denote the region exterior to both balls. Since Hilbert balls are convex, $p$, $q$, and $r$ all lie on the boundary of $E$. It follows that for any point $w$ in $E$, there exist three nonintersecting curves joining $v$ to $p$, $q$, and $r$. Since both balls are empty, points $c_1$ and $c_2$ are vertices in the Voronoi diagram of $\{p, q, r\}$. Since Voronoi cells are star-shaped, the six open line segments joining each center to each point do not intersect. Further, these lie within $B_1 \cup B_2$, and so they are disjoint from the curves joining $w$ to these points. Therefore the complete bipartite graph $\{p,q,r\} \times \{c_1, c_2, w\}$ is planar, but this is impossible since it would imply a planar embedding of $K_{3,3}$ (the complete bipartite $3 \times 3$ graph). 
\end{proof}

Note that the non-collinear assumption is necessary. To see why, let $\Omega$ be an axis aligned rectangle, and let $\ell$ be a horizontal line through the rectangle's center. It is easy to construct two Hilbert balls, one above the line and one below, whose boundaries contain a segment along this line. Placing the three points within this segment would violate the lemma. 

The following lemma shows that $\DT(P)$ is a connected, planar graph. Its proof is similar to the Euclidean case.

\begin{restatable}{lemma}{planarSpanner} \label{lem:planar-spanner}
Given any discrete set of points $P$, $\DT(P)$ is a planar (straight-line) graph that spans $P$.
\end{restatable}

\begin{proof}
Suppose to the contrary that there existed $a,b,c,d \in P$ such that the edges $(a,b)$ and $(c,d)$, viewed as open line segments, intersected each other. These edges are witnessed by two empty Hilbert balls, $B_1$ for $(a,b)$ and $B_2$ for $(c,d)$ (see Figure~\ref{fig:planar-spanner}(a)). By convexity, $B_1$ contains the segment $a b$, and by emptiness, its interior does not contain $c$ or $d$. Similarly, $B_2$ contains the segment $c d$, and its interior does not contain $a$ or $b$. and $B_2$. It follows the boundaries of $B_1$ and $B_2$ must intersect in at least four points. Applying Lemma~\ref{lem:circum-unique} to any three of these four intersection points yields a contradiction.

\begin{figure}[htbp]
    \centerline{\includegraphics[scale=0.40]{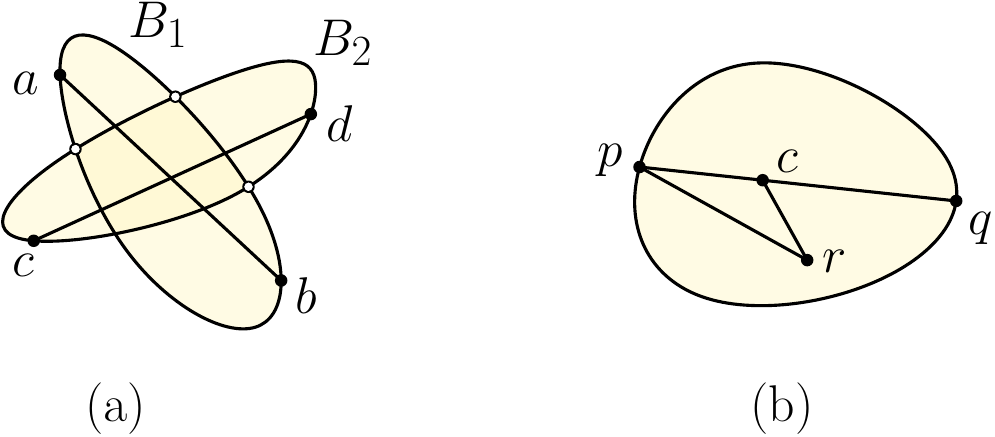}}
    \caption{Proof of Lemma~\ref{lem:planar-spanner}.} 
    \label{fig:planar-spanner}
\end{figure}

To prove the $\DT(P)$ spans $P$, suppose to the contrary that this was not so. Consider the two sites $p,q \in P$ that are in different components but closest in terms of their Hilbert distance. Let $c$ denote the center of the smallest Hilbert ball $B$ containing $p$ and $q$. Because straight lines are geodesics, $c$ lies on the line segment $p q$, and $d_{\Omega}(p,q) = d_{\Omega}(p,c) + d_{\Omega}(c,q)$ (see Figure~\ref{fig:planar-spanner}(b)). We assert that this ball can contain no other site in its interior. If there were such a site $r$, it lies in a different connected component from $p$ or $q$ (possibly both). Let us assume without loss of generality that it is not in $p$'s connected component. Because it is in the ball's interior, we have $d_{\Omega}(c,r) < d_{\Omega}(c,q)$. By the triangle inequality, we obtain 
\[
    d_{\Omega}(p,r) 
        ~ \leq ~ d_{\Omega}(p,c) + d_{\Omega}(c,r) 
        ~ <    ~ d_{\Omega}(p,c) + d_{\Omega}(c,q)
        ~ =    ~ d_{\Omega}(p,q),
\]
but this contradicts the hypothesis that $p$ and $q$ were the closest sites in different components.
\end{proof}

Another useful property of the Delaunay triangulation, viewed as a graph, is its relationship to other geometric graph structures. The \emph{Hilbert minimum spanning tree} of a point set $P$, denoted $\text{MST}_{\Omega}(P)$, is the spanning tree over $P$ where edge weights are Hilbert distances. The \emph{Hilbert relative neighborhood graph} for a point set $P$, denoted $\text{RNG}_{\Omega}(P)$ is a graph over the vertex set $P$, where two points $p,q \in P$ are joined by an edge if $d_{\Omega}(p,q) \leq \min(d_{\Omega}(p,r), d_{\Omega}(q,r))$, for all $r \in P \setminus \{p,q\}$. Toussaint proved that in the Euclidean metric, the minimum spanning tree is contained in the relative neighborhood graph, which is contained in the Delaunay graph~\cite{toussaint1980rng}. The following shows that this holds in the Hilbert metric as well.

\begin{restatable}{lemma}{mstRng} \label{lem:mst-rng}
Given any discrete set of points $P$, $\text{MST}_{\Omega}(P) \subseteq \text{RNG}_{\Omega}(P) \subseteq \DT_{\Omega}(P)$.
\end{restatable}

\begin{proof}
The fact that the minimum spanning tree is contained in the relative neighborhood graph holds in any metric. If two points $p$ and $q$ are not joined by an edge in the RNG, then there is a point $r$ that is closer to both than they are to each other. The edge $p q$ cannot be in the minimum spanning tree, because it is the highest weight edge in the cycle formed by $p$, $q$, and $r$.

To prove that the $\text{RNG}_{\Omega}(P) \subseteq \DT_{\Omega}(P)$, consider an edge $p q$ in the Hilbert RNG. Let $B(p,q)$ denote the Hilbert ball centered at $p$ with radius $d_{\Omega}(p,q)$ and define $B(q,p)$ analogously. The condition for two points $p,q \in P$ to be connected in $\text{RNG}_{\Omega}(P)$ is that there is no other point of $P$ lying in the interior of ``lune'' $B(p,q) \cap B(q,p)$. Consider the ball $B$ centered at the point $c$ that is midway between $p$ and $q$, as introduced in the proof of Lemma~\ref{lem:planar-spanner}. For any $r \in B$, we showed that $d_{\Omega}(p,r) < d_{\Omega}(p,q)$. Symmetrically, we have $d_{\Omega}(q,r) < d_{\Omega}(q,p)$, which implies that $r$ lies in the lune. Since the lune must be empty, $B$ must also be empty. This empty ball witnesses the fact that $p q$ is an edge of $\DT_{\Omega}(P)$.
\end{proof}

\section{Hilbert Bisectors} \label{sec:hilbert-bisect}

In this section we present some utilities for processing Hilbert bisectors, which will be used later in our algorithms. We start by recalling a characterization given by Gezalyan and Mount for when a point lies on the Hilbert bisector~\cite{gezalyan2021voronoi}. Recall that three lines are said to be \emph{concurrent} in projective geometry if they intersect in a common point or are parallel.

\begin{lemma} \label{lem:bisector-prop}
Given a convex body $\Omega$ and two points $p,q \in \interior(\Omega)$, consider any other point $x \in \Omega$. Let $p'$ and $p''$ denote the endpoints of the chord $\overline{p x}$ so the points appear in the order $\ang{p', p, x, p''}$. Define $q'$ and $q''$ analogously for $\overline{q x}$.
\begin{enumerate}
\item[$(i)$] If $x \in \interior(\Omega)$, then it lies on the $(p,q)$-bisector if and only if the lines $\overleftrightarrow{p q\strut}$, $\overleftrightarrow{p' q'\strut}$, and $\overleftrightarrow{p'' q''\strut}$ are concurrent (see Figure~\ref{fig:bisector-prop}(a)). 

\item[$(ii)$] If $x \in \bd \Omega$ (implying that $x = p'' = q''$), then $x$ is the limit point of the $(p,q)$-bisector if and only if there is a supporting line through $x$ concurrent with $\overleftrightarrow{p q\strut}$ and $\overleftrightarrow{p' q'}$ (see Figure~\ref{fig:bisector-prop}(b)).
\end{enumerate}
\end{lemma}

\begin{figure}[htbp]
    \centerline{\includegraphics[scale=0.40]{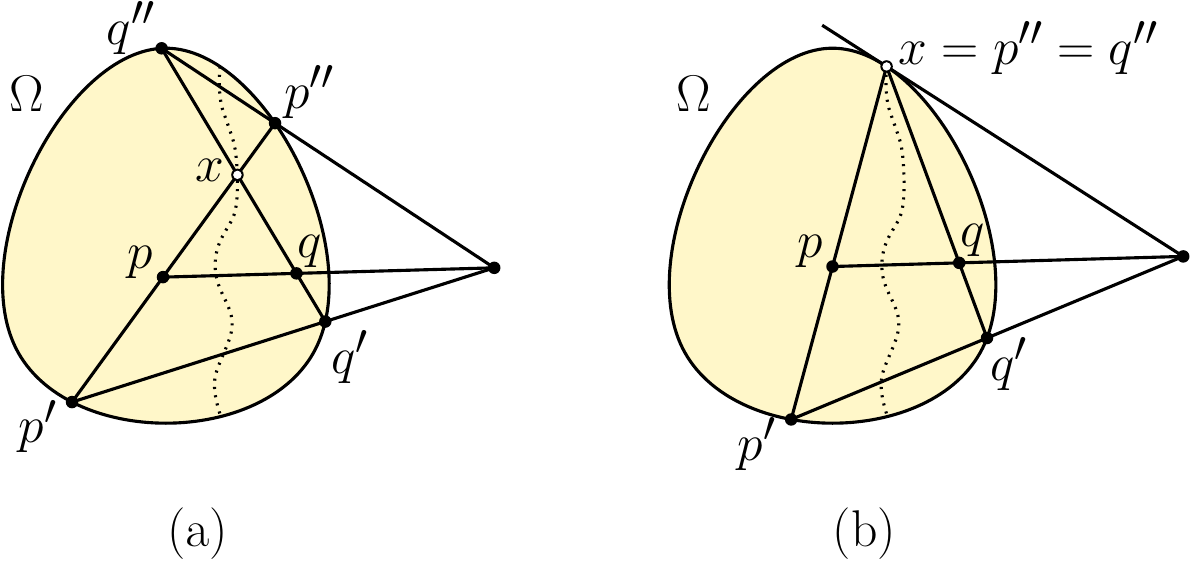}}
    \caption{Properties of points on the Hilbert bisector.}
    \label{fig:bisector-prop}
\end{figure}

The following utility lemmas arise as consequences of this characterization. Both involve variants of binary search.

\begin{restatable}{lemma}{bisectorRayHit} \label{lem:bisector-ray-hit}
Given an $m$-sided convex polygon $\Omega$, two points $p,q \in \interior(\Omega)$, and any ray emanating from $p$, the point of intersection between the ray and the $(p,q)$-bisector (if it exists) can be computed in $O(\log^2 m)$ time.
\end{restatable}

\begin{proof}
Through binary search on the boundary of $\Omega$, in $O(\log m)$ time we can determine the endpoints $p'$ and $p''$ of the  chord passing through $p$ along the ray (see Figure~\ref{fig:bisector-ray-hit}(a)). Let $\ell$ denote this chord, and let $V$ denote the points along the chord where the spokes of $q$ intersect $\ell$ (see Figure~\ref{fig:bisector-ray-hit}(b)). 

\begin{figure}[htbp]
    \centerline{\includegraphics[scale=0.40]{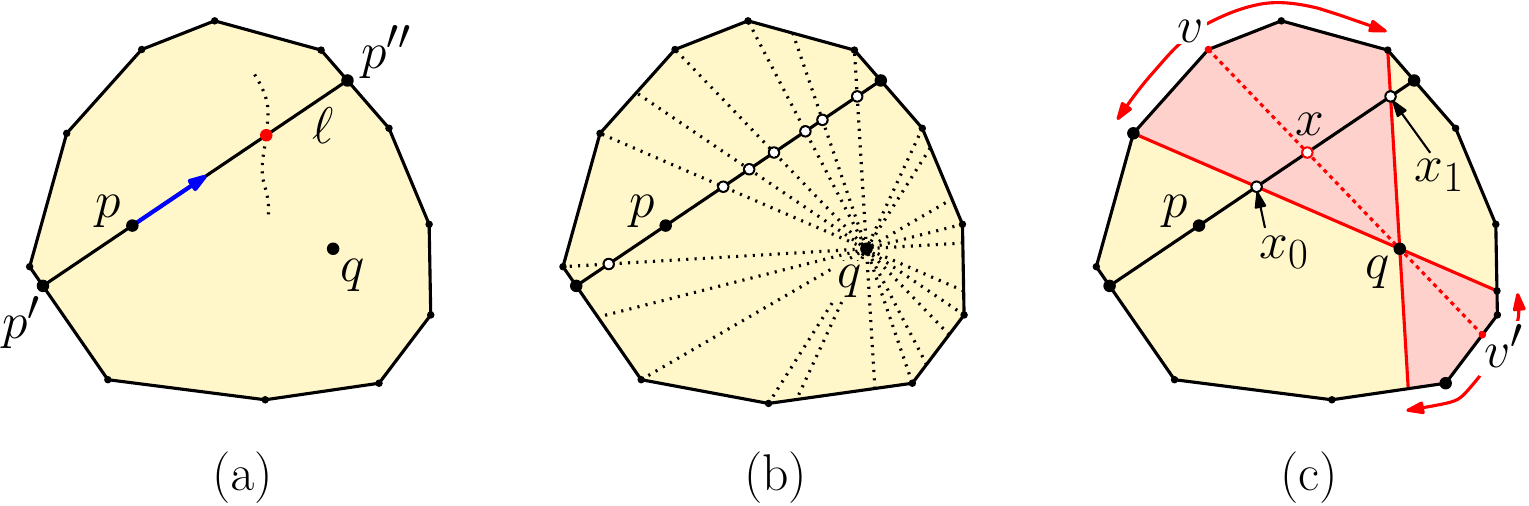}}
    \caption{Computing where a ray through a site hits the bisector.}
    \label{fig:bisector-ray-hit}
\end{figure}

While the points of $V$ will not be computed explicitly, the search algorithm will simulate a binary search among these points. Observe that for any subsegment $[x_0,x_1] \subseteq \ell$, the two chords $\overline{q x_0}$ and $\overline{q x_1}$ subdivide the boundary of $\Omega$ into two polygonal arcs so that every chord through $q$ that intersects this subsegment has its endpoints in these two arcs. We can effectively simulate a binary search along $[x_0,x_1]$ by sampling the median vertex $v$ from whichever arc contains the greater number of vertices, and then taking the point $x = \overline{q v} \cap \ell$ (see Figure~\ref{fig:bisector-ray-hit}(c)). Each such probe will succeed in eliminating at least one quarter of the remaining points of $V$.

Given a probe point $x$ along $\ell$, we first determine in $O(1)$ time whether it is on the proper side of $p$ relative to the ray's direction. If not, we eliminate the portion lying behind $p$. If so, we compute the opposite endpoint $v'$ of $\overline{q v}$ in $O(\log m)$ time. Next, we compute the distances $d_{\Omega}(p,x)$ and $d_{\Omega}(q,x)$ (which can be done in $O(1)$ time since both chords are known). If the distance to $q$ is smaller, we recurse on $[x_0, x]$, and otherwise we recurse on $[x, x_1]$. In either case, a constant fraction of remaining points have been eliminated. 

The search eventually terminates either with an interval along $\ell$ between two consecutive spokes of $q$, or it falls off the end of the ray. In the former case, we can determine the exact location of the bisector within the interval in $O(1)$ time, and report this answer. In the latter case, we report that the ray does not intersect the bisector.

Each probe takes $O(\log m)$ time (to determine the opposite endpoint of each ray through $q$). Since $V$ contains $O(m)$ points, and we eliminate a constant fraction with each probe, the overall time is $O(\log^2 m)$, as desired.
\end{proof}

\begin{restatable}{lemma}{bisectorEndpoint} \label{lem:bisector-endpoint}
Given an $m$-sided convex polygon $\Omega$ and any two points $p,q \in \interior(\Omega)$, the endpoints of the $(p,q)$-bisector on $\bd \Omega$ can be computed in $O(\log^2 m)$ time.
\end{restatable}

\begin{proof}
For the sake of illustration, let us assume that the line through $p$ and $q$, denoted $\ell$, is horizontal (see Figure~\ref{fig:bisector-endpoint}(a)). In $O(\log m)$ time we can compute the two points, $z_0$ and $z_1$, where $\ell$ intersects $\bd \Omega$, with $z_0$ on the left side of $\Omega$ and $z_1$ on the right. These points implicitly split the boundary of $\Omega$ into two convex chains, which we call the \emph{upper} and \emph{lower chains}. We will show how to compute the endpoint of the $(p,q)$-bisector that lies on the upper chain, and the other case is symmetrical. 

\begin{figure}[htbp]
    \centerline{\includegraphics[scale=0.40]{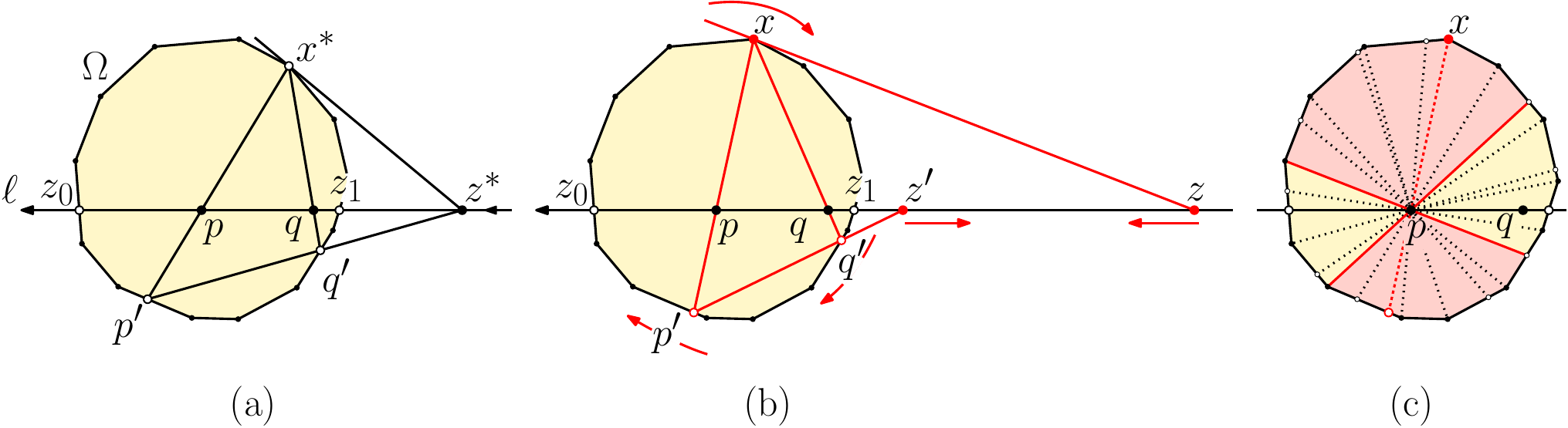}}
    \caption{Computing the endpoints of the $(p,q)$-bisector.}
    \label{fig:bisector-endpoint}
\end{figure}

It will be convenient to think of the portion of $\ell$ that lies outside of $\Omega$ as being oriented from right to left, first starting at $z_0$, then progressing to left to $-\infty$ and wrapping around to $+\infty$, and finally progressing to $z_1$. Observe that as a point $x$ travels from $z_0$ to $z_1$ along the upper chain of $\bd \Omega$, the associated supporting line intersects $\ell$ at a point that moves monotonically along $\ell$ from $z_0$ to $z_1$ (wrapping around when it hits the topmost vertex of $\Omega$).

By Lemma~\ref{lem:bisector-prop}(ii), we seek a point $x^*$ on the upper chain that satisfies the conditions of the lemma (see Figure~\ref{fig:bisector-endpoint}(a)). Given any point $x$ on the upper chain, consider the endpoints $p'$ and $q'$ of the respective chords $\overline{x p}$ and $\overline{x q}$ (see Figure~\ref{fig:bisector-endpoint}(b)). Let $z$ denote the point where a supporting line at $x$ intersects $\ell$. Let $z'$ denote the point where the line $\overleftrightarrow{p' q'}$ intersects $\ell$. In the limit, when $x \rightarrow z_0$, we have $z = z_0$ and $z' = z_1$, and in the other limit, when $x \rightarrow z_1$, we have $z = z_1$ and $z' = z_0$. The key to the search process is that as $x$ moves monotonically along the upper chain from $z_0$ to $z_1$, the point $z$ moves monotonically along $\ell$ from $z_0$ to $z_1$. On the other hand, the points $p'$ and $q'$ move monotonically along the lower chain from $z_1$ and $z_0$, and hence $z'$ moves in the reverse direction along $\ell$. It follows that there is a unique point $x$ on the upper chain where $z$ and $z'$ coincide along $\ell$. This will be the desired point $x^*$.

We will show that $x^*$ can be found through a binary search along the upper chain of $\Omega$, where each probe takes $O(\log m)$ time. We maintain two spokes through $p$, which delimit the portion along the upper chain where the desired point $x^*$ resides (see Figure~\ref{fig:bisector-endpoint}(c)). This partitions the upper and lower chain into two polygonal arcs containing the endpoints of the candidate chords. With each probe, we sample the median vertex $x$ from the arc that contains the larger number of vertices. This will allow us to eliminate at least one quarter of the remaining vertices from consideration. 

Depending on the nature of the spoke, the point $x$ may be a vertex of the upper chain or it may lie in the interior of an edge. If it is in the interior of an edge, the supporting line at $x$ is uniquely determined, implying that $z$ is uniquely determined as well. In $O(\log m)$ time, we can determine the points $p'$ and $q'$, and $z'$ can be computed in $O(1)$ additional time. Depending on where $z'$ lies relative to $z$ in the order defined along $\ell$, we eliminate either the region before or after $x$ on the upper chain, thus eliminating a constant fraction of the points from further consideration. On the other hand, if $x$ is a vertex, there are two supporting lines, which define an interval along $\ell$. Again, we can determine $p'$, $q'$, and $z'$ in $O(\log m)$ time. If $z'$ lies within this interval, then $x$ is the desired final answer. Otherwise, we eliminate either the region before or after $x$, depending on which side of the interval $z'$ lies.

Each probe takes $O(\log m)$ time (to determine the points $p'$ and $q'$). Since there are $O(m)$ spokes through $p$, and we eliminate a constant fraction with each probe, the overall time is $O(\log^2 m)$, as desired.
\end{proof}

\section{Hilbert Circumcircles} \label{sec:hilbert-circum}

In the Euclidean plane, any triple of points that are not collinear lie on a unique circle, that is, the boundary of an Euclidean ball. This is not true in the Hilbert geometry, however. Since Delaunay triangulations are based on an empty circumcircle condition, it will be important to characterize when a triangle admits a Hilbert circumcircle and when it does not. In this section we explore the conditions under which such a ball exists.

\subsection{Balls at infinity} \label{sec:ball-at-infty}

We begin by introducing the concept of a Hilbert ball centered at a point on the boundary of $\Omega$. Let $x$ be any point on $\bd \Omega$. Given any point $p \in \interior(\Omega)$, we are interested in defining the notion of a Hilbert ball centered at $x$ whose boundary contains $p$. If we think of the points on the boundary of $\Omega$ as being infinitely far from any point in $\interior(\Omega)$, this gives rise to the notion of a ball at infinity.

Given $x \in \bd \Omega$ and $p \in \interior(\Omega)$, let $u$ be any unit vector directed from $x$ into the interior of $\Omega$ (see Figure~\ref{fig:ball-at-infty}(a)). Given any sufficiently small positive $\delta$, let $x_{\delta} = x + \delta u$ denote the point at distance $\delta$ from $x$ along vector $u$, and let $B(x_{\delta}, p)$ denote the Hilbert ball centered at $x_{\delta}$ of radius $d_{\Omega}(x_{\delta}, p)$. The following lemma shows that as $\delta$ approaches $0$, this ball approaches a shape, which we call the \emph{Hilbert ball at infinity} determined by $x$ and $u$ and passing through $p$, denoted $B_u(x, p)$ (see Figure~\ref{fig:ball-at-infty}(b)).

\begin{figure}[htbp]
    \centerline{\includegraphics[scale=0.40]{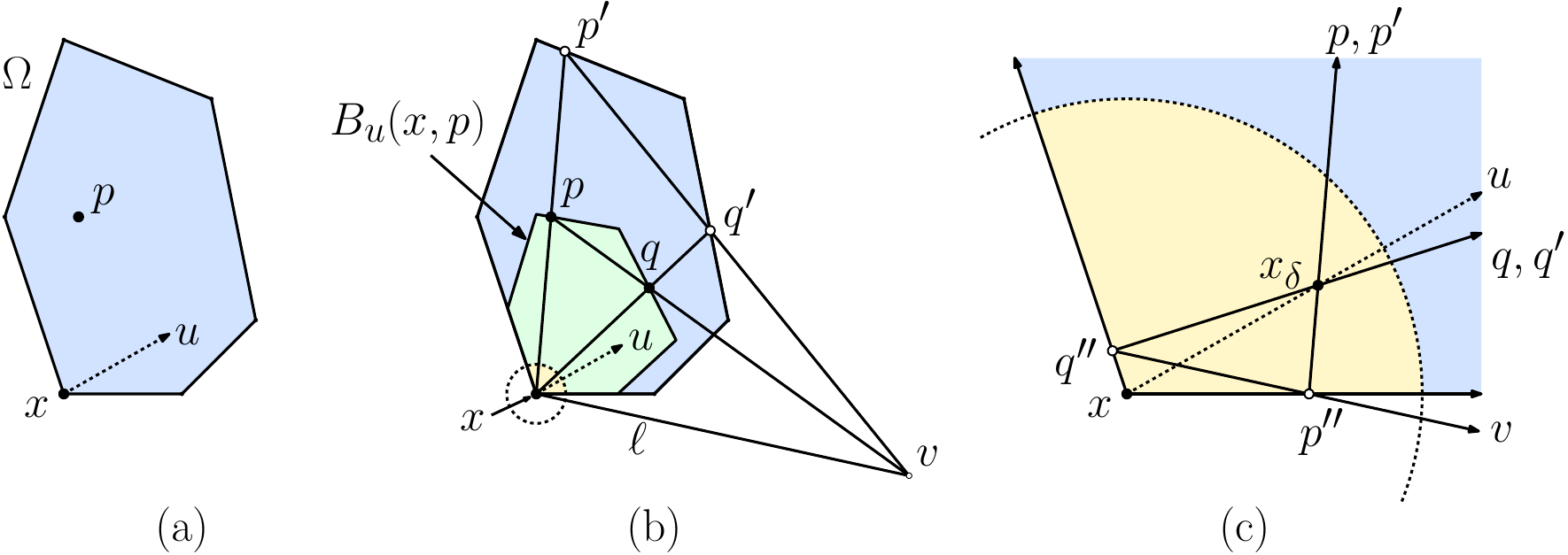}}
    \caption{Constructing the Hilbert ball $B_u(x, p)$ at infinity.} 
    \label{fig:ball-at-infty}
\end{figure}

\begin{restatable}{lemma}{ballAtInfty} \label{lem:ball-at-infty}
As $\delta$ approaches $0$, $B(x_{\delta}, p)$ approaches a convex polygon lying within $\Omega$ having $x$ on its boundary.
\end{restatable}

\begin{proof}
We will provide a construction of the limiting shape by deriving the point at which any chord through $x$ passes through the boundary of the shape. First, consider the chord $\overline{p x_{\delta}}$, letting $p'$ and $p''$ denote the two points of intersection with $\bd \Omega$. Next, take any point $q' \in \bd \Omega$ different from $p'$. Consider the chord $\overline{q' x_{\delta}}$, letting $q''$ denote the other endpoint of the chord. Label the chord endpoints so that $p'$ and $q'$ are on the same side as $p$ with respect to $x_{\delta}$ and $p''$ and $q''$ are on the other (see Figure~\ref{fig:ball-at-infty}(b) and (c)). Let $v$ denote the point where the lines $\overleftrightarrow{p'q'}$ and $\overleftrightarrow{p''q''}$ intersect. Let $q$ denote the point of intersection between the chord $\overline{q'q''}$ and the line $\overleftrightarrow{p v}$. (It is easy to see that $q$ exists and is unique by convexity.)

Observe that the point sequences $\ang{p'', x_{\delta}, p, p'}$ and $\ang{q'', x_{\delta}, q, q'}$ form a projectivity through $v$, and hence their cross ratios are equal. Therefore, $p$ and $q$ are equidistant from $x_{\delta}$, and therefore they both lie on the boundary of a  Hilbert ball centered at $x_{\delta}$. 

All of these quantities depend on the value of $\delta$. Observe that in the limit, as $\delta \rightarrow 0$, $x_{\delta}$ approaches $x$, and thus the chords $x_{\delta} p$ and $x_{\delta} q'$ approach $x p$ and $x q'$, respectively. Also observe that for all sufficiently small $\delta$, the edges on which $p''$ and $q''$ lie do not vary. Therefore, for these and all smaller values of $\delta$, the triangles $\triangle x p'' q''$ are similar, and hence the slope of the line $\overleftrightarrow{p'' q''}$ is fixed. (Fixing this slope is the reason that the limiting ball depends not only on $x$ but on $u$ as well.) It follows that this line approaches the supporting line $\ell$ through $x$ with this same slope. (Figure~\ref{fig:ball-at-infty}(c) shows the case where $x$ is a vertex. When $x$ lies on the interior of an edge, the choice of $u$ does not matter, and the supporting line is the unique supporting line through this edge.) 

Therefore, the point where lines $\overleftrightarrow{p' q'}$ and $\ell$ intersect approaches a limiting point $v$. It follows that $q$ approaches a point along the chord $x q'$ in the limit. The set of all such points taken over all choices of $q'$ on the boundary of $\Omega$ yields the boundary of the desired shape. Since all the intermediate balls $B(x_{\delta}, p)$ are convex polygons lying within $\Omega$ containing $x_{\delta}$, it follows that the limiting shape is also a convex polygon containing $x$, which completes the proof.
\end{proof}

A useful utility, which we will apply in Section~\ref{sec:hilbert-hull}, involves computing the largest empty ball centered at any boundary point $x$ with respect to a point set $P$.

\begin{restatable}{lemma}{emptyBallAtInfty} \label{lem:empty-ball-at-infty}
Given a set of $n$ points $P$ in the interior of $\Omega$, and any point $x \in \bd \Omega$, in $O(n \log m)$ time, it is possible to compute a point $p \in P$, such that there is a ball at infinity centered at $x$ whose boundary passes through $p$, and which contains no points of $P$ in its interior.
\end{restatable}

\begin{proof}
Let $P = \{p_1, \ldots, p_n\}$, and let $\ell$ be any supporting line passing through $x$, and fix an arbitrary point $y \in \bd \Omega$ other than $x$ (see Figure~\ref{fig:empty-ball-at-infty}(a)). For each point $p \in P$, shoot a ray from $x$ through $p$, letting $p'$ denote the point where it hits the boundary. Let $v$ denote the intersection of $\ell$ and the line $\overleftrightarrow{p' y}$, and let $q$ denote the intersection of line $\overleftrightarrow{p v}$ with the chord $\overline{x y}$ (see Figure~\ref{fig:empty-ball-at-infty}(b)). 

\begin{figure}[htbp]
     \centerline{\includegraphics[scale=0.40]{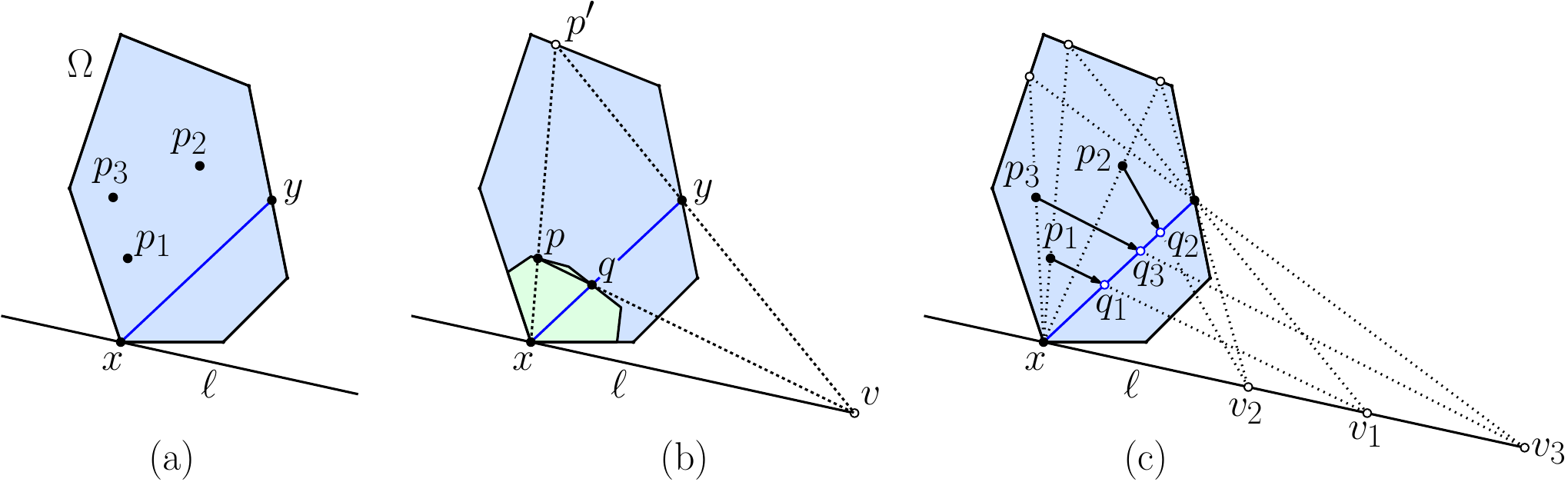}}
    \caption{Proof of Lemma~\ref{lem:empty-ball-at-infty}.} 
    \label{fig:empty-ball-at-infty}
\end{figure}

It follows from Lemma~\ref{lem:bisector-prop} that both $p$ and $q$ lie on the boundary of a ball at infinity centered at $x$. Applying this construction to each $p_i \in P$ yields a set of points $\{q_1, \ldots, q_n\}$ along $\overline{x y}$, which by Lemma~\ref{lem:circum-unique} are ordered by inclusion (see Figure~\ref{fig:empty-ball-at-infty}(c)). The desired point $p_i \in P$ is the one whose corresponding point $q_i$ is closest to $x$ along this chord. The running time for each point $p_i$ is $O(\log m)$ to determine where the ray hits the boundary, which yields an overall running time of $O(n \log m)$.
\end{proof}

Balls at infinity are not proper aspects of Hilbert geometry, but they will be convenient for our purposes. Given two points $p, q \in \interior(\Omega)$, let $x$ denote the endpoint of the $(p,q)$-bisector on $\bd \Omega$, oriented so that $x$ lies to the left of the directed line $\overrightarrow{p q}$. Let $u$ denote the tangent vector of the bisector at $x$. Define $B(p {\ST} q) = B_u(x, p)$. Note that this (improper) ball is both centered at and passes through $x$. In this sense it circumscribes the triangle $\triangle p q x$. Define $B(q {\ST} p)$ analogously for the opposite endpoint of this bisector (see Figure~\ref{fig:overlap-outer}(a) and~(b)).

We can now characterize the set of points $r$ that admit a Hilbert circumcircle with respect to two given points $p$ and $q$. This characterization is based on two regions, called the \emph{overlap} and \emph{outer regions} (see Figure~\ref{fig:overlap-outer}(c)).

\begin{definition}[Overlap/Outer Regions]
Given two points $p, q \in \interior(\Omega)$:
\smallskip
\begin{description}
\item[Overlap Region:] denoted $Z(p,q)$, is $B(p {\ST} q) \cap B(q {\ST} p)$.

\item[Outer Region:] denoted $W(p,q)$, is $\Omega \setminus (B(p {\ST} q) \cup B(q {\ST} p))$.
\end{description}
\end{definition}

\begin{figure}[htbp]
    \centerline{\includegraphics[scale=0.40]{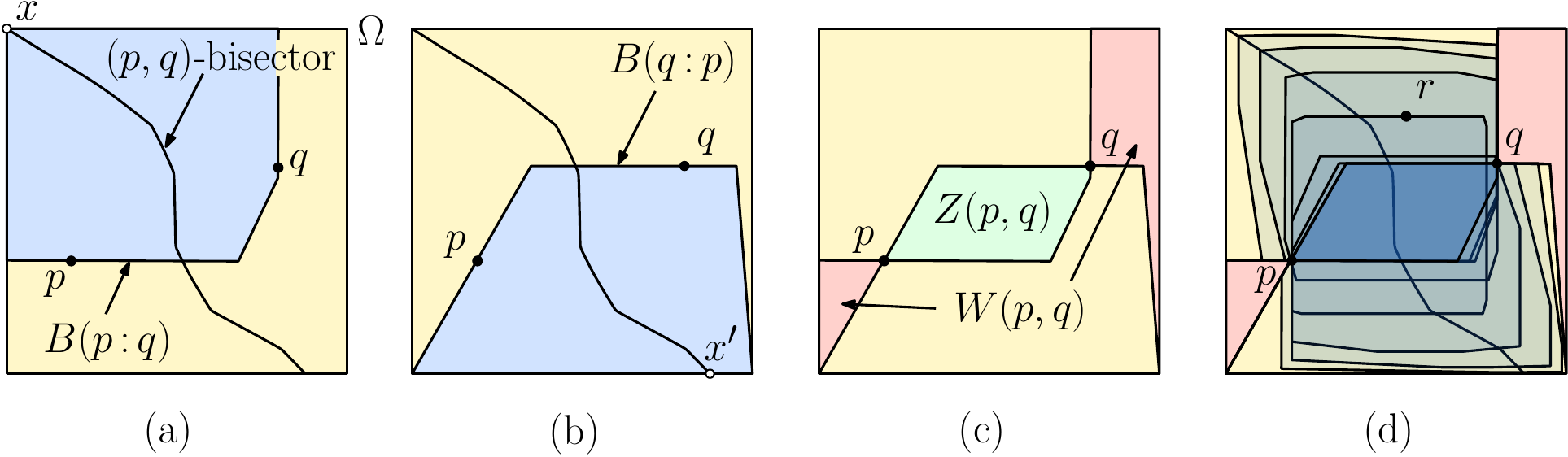}}
    \caption{The overlap region $Z(p,q)$ and the outer region $W(p,q)$.}
    \label{fig:overlap-outer}
\end{figure}

\begin{restatable}{lemma}{forbidden} \label{lem:forbidden}
A triangle $\triangle p q r \subseteq \interior(\Omega)$ admits a Hilbert circumcircle if and only if $r \notin Z(p,q) \cup W(p,q)$. 
\end{restatable}

\begin{proof}
Let $x,x' \in \bd \Omega$ denote the endpoints of the $(p,q)$-bisector. Since the bisector is a curve, we can parameterize points along this curve by $s:[0,1]\rightarrow \Omega$ so that $s(0) = x$ and $s(1) = x'$. For each $t \in (0,1)$ there exists a unique Hilbert ball centered at $s(t)$ and passing through $p$ and $q$, denoted $B_t(p {\ST} q)$. Lemma~\ref{lem:circum-unique} implies that any point $r$ other than $p$ and $q$ that lies on the boundary of one of these circles, can lie on at most one such circle. That is, the set $B_t(p {\ST} q)$ defines a ``pencil'' of Hilbert circles (see Figure~\ref{fig:overlap-outer}(d)) ranging from $B(p {\ST} q) = B_0(p {\ST} q)$ to $B(q {\ST} p) = B_1(p {\ST} q)$. Conversely, any Hilbert circumcircle passing through $p$ and $q$ must be generated by a point on the $(p,q)$-bisector, and therefore every point $r$ that completes a Hilbert circumcircle lies on one of the circles in this pencil. It is evident therefore that $r$ cannot lie within either $Z(p,q)$ nor $W(p,q)$, for otherwise the boundary of the resulting ball would cross the boundary of either $B(p {\ST} q)$ or $B(q {\ST} p)$, thus violating Lemma~\ref{lem:circum-unique} at this crossing point.
\end{proof}

As shown in Figure~\ref{fig:voronoi-delaun}, the Delaunay triangulation need not cover the convex hull of the set of sites. We refer to the region that is covered as the \emph{Hilbert hull} of the sites. Later, we will present an algorithm for computing the Hilbert hull. The following lemma will be helpful. It establishes a nesting property for the overlap regions.

\begin{restatable}{lemma}{zNesting} \label{lem:z-nesting}
If $r \in Z(p,q)$ then $Z(p,r) \subset Z(p,q)$.
\end{restatable}

\begin{proof}
Given two points $p$ and $q$, let $r$ be a point in $Z(p,q)$. Consider that the bisector between $p$ and $r$ and the bisector between $q$ and $r$ split $\Omega$ into three pieces that can only connect at infinity. Let $u$ be a point in $Z(p,r)$. Then this would necessarily require that the bisectors between $p$ and $u$ and $u$ and $r$ would split $\Omega$ again such that they only can connect at infinity. Note that every point on the $u,r$ bisector is necessarily closer to $u$ than to $r$ and therefore $q$. Hence $u$ must split space between $p$ and $q$ as well. Since $u$ was defined to be in one of the balls at infinity it must therefore by in $Z(p,q)$ as it cannot lie in $W(p,q)$.
\end{proof}

\section{Computing Circumcircles} \label{sec:compute-circum}

A fundamental primitive in the Euclidean Delaunay triangulation algorithm is the so-called \emph{in-circle test}~\cite{guibas1992randomized}. Given a triangle $\triangle p q r$ and a fourth site $s$, the test determines whether $s$ lies within the circumscribing Hilbert ball for the triangle (if such a ball exits). In this section, we present an algorithm which given any three sites either computes the Hilbert circumcircle for these sites, denoted $B(p {\ST} q {\ST} r)$, or reports that no circumcircle exists. The in-circle test reduces to checking whether $d_{\Omega}(s, c) < \rho$, which can be done in $O(\log m)$ time 
as observed in Section~\ref{sec:prelim}.
%

\begin{lemma} \label{lem:compute-circum}
Given a convex $m$-sided polygon $\Omega$ and triangle $\triangle p q r \subset \interior(\Omega)$, in $O(\log^3 m)$ time it is possible to compute $B(p {\ST} q {\ST} r)$ or to report that no ball exists.
\end{lemma}

The remainder of the section is devoted to the proof. The circumscribing ball exists if and only if there is a point equidistant to all three, implying that $(p,q)$- and $(p,r)$-bisectors intersect at some point $c \in \interior(\Omega)$. The following technical lemma shows that bisectors intersect crosswise.

\begin{restatable}{lemma}{pseudoline} \label{lem:pseudoline}
Given three non-collinear points $p, q, r \in \interior(\Omega)$, the $(p,q)$- and $(p,r)$-bisectors have endpoints lying on $\bd \Omega$. If they intersect within $\interior(\Omega)$, they intersect transversely in a single point (see Figure~\ref{fig:circumcircle-1}).
\end{restatable}

\begin{proof}
As shown in \cite{gezalyan2021voronoi}, star-shapedness of Voronoi cells implies that bisectors are not closed loops. Since each bisector is a simple curve, a bisector has two endpoints on the boundary of $\Omega$. If the $(p,q)$- and $(p,r)$-bisectors intersect, the intersection point is equidistant from all three, and hence is the center of a circumcircle. By the above lemma implies that there cannot be more than one intersection. The intersection must be transverse. If not, the regions on opposite sides of the intersection point would be closer to $p$ than to $q$ or $r$, and this would violate the property the Voronoi cells are star-shaped.
\end{proof}

\begin{figure}[htbp]
    \centerline{\includegraphics[scale=0.40]{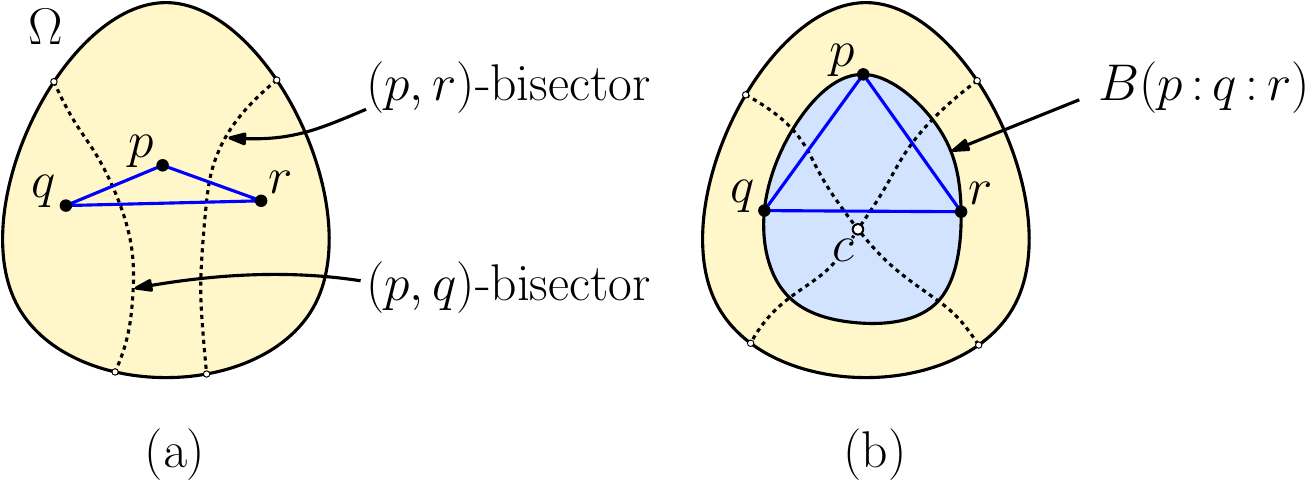}}
    \caption{Circumcircles and bisector intersection.}
    \label{fig:circumcircle-1}
\end{figure}

It follows that the $(p,q)$- and $(p,r)$-bisectors subdivide the interior of $\Omega$ into either three or four regions (three if they do not intersect, and four if they do). We can determine which is this case by invoking Lemma~\ref{lem:bisector-endpoint} to compute the endpoints of these bisectors in $O(\log^2 m)$ time. If they alternate between $(p,q)$ and $(p,r)$ along the boundary of $\Omega$, then the bisectors intersect, and otherwise they do not. In the latter case, there is no circumcircle for $p$, $q$, and $r$, and hence, no possibility of violating the circumcircle condition. (Note that this effectively provides an $O(\log^2 m)$ test for Lemma~\ref{lem:forbidden}.) Henceforth, we concentrate on the former case.

Let us assume, without loss of generality, that $\triangle p q r$ is oriented counterclockwise. Let $v_q$ denote the endpoint of the $(p,q)$-bisector that lies to the right of the oriented line $\overrightarrow{p q}$. Let $v_r$ denote the endpoint of the $(p,r)$-bisector that lies to the left of the oriented line $\overrightarrow{p r}$ (see Figure~\ref{fig:circumcircle-2}(a)). 

\begin{figure}[htbp]
    \centerline{\includegraphics[scale=0.40]{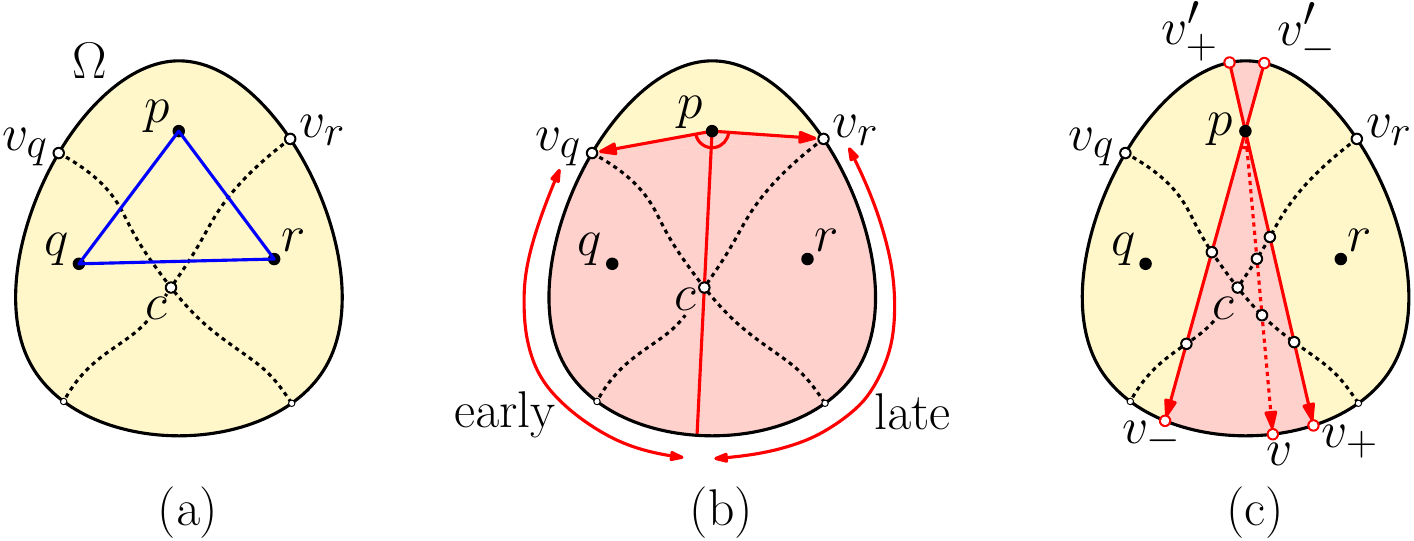}}
    \caption{Computing the center of $\triangle{p q r}$.}
    \label{fig:circumcircle-2}
\end{figure}

Because Voronoi cells are star-shaped, it follows that the vector from $p$ to the desired circumcenter $c$ lies in the counterclockwise angular interval from $\overrightarrow{p v_q}$ to $\overrightarrow{p v_r}$ (see Figure~\ref{fig:circumcircle-2}(b)). The key to the search is the following observation. In the angular region from $\overrightarrow{p v_q}$ to $\overrightarrow{p c}$, any ray shot from $p$ intersects the $(p,q)$-bisector before hitting the $(p,r)$-bisector (if it hits the $(p,r)$-bisector at all). On the other hand, in the angular region from $\overrightarrow{p c}$ to $\overrightarrow{p v_r}$, any ray shot from $p$ intersects the $(p,r)$-bisector before hitting the $(p,q)$-bisector (if it hits the $(p,q)$-bisector at all). We say that the former type of ray is \emph{early} and the latter type is \emph{late}.

Let $v_{-}$ and $v_{+}$ denote the points on $\bd \Omega$ that bound the current search interval about $p$ (see Figure~\ref{fig:circumcircle-2}(c)). We will maintain the invariant that the ray $\overrightarrow{p v_{-}}$ is early and $\overrightarrow{p v_{+}}$ is late. Initially, $v_{-} = v_q$ and $v_{+} = v_r$. Let $v'_{-}$ and $v'_{+}$ denote the opposite endpoints of the chords $\overline{p v_{-}}$ and $\overline{p v_{+}}$. Consider the portion of the boundary of $\Omega$ that lies in the counterclockwise from $v_{-}$ and $v_{+}$. If the angle $\angle v_{-} p v_{+}$ is smaller than $\pi$, also consider the portion of the boundary of $\Omega$ that lies in the counterclockwise from $v'_{-}$ and $v'_{+}$. With each probe, we sample the median vertex $v$ from whichever of these two boundary portions that contains the larger number of vertices. If $v$ comes from the interval $[v_{-},v_{+}]$, we probe along the ray $\overrightarrow{p v}$, and if it comes from the complementary interval, $[v'_{-},v'_{+}]$, we shoot the ray from $p$ in the opposite direction from $v$. We then apply Lemma~\ref{lem:bisector-ray-hit} twice to determine in $O(\log^2 m)$ time where this ray hits the $(p,q)$- and $(p,r)$-bisectors (if at all). Based on the results, we classify this ray as being early or late and recurse on the appropriate angular subinterval. When the search terminates, we have determined a pair of consecutive spokes about $p$ that contain $c$. 

Because each probe eliminates at least half of the vertices from the larger of the two boundary portions, it follows that at $O(\log m)$ probes, we have located $c$ to within a single pair of consecutive spokes around $p$. Since each probe takes $O(\log^2 m)$ time, the entire search takes $O(\log^3 m)$ time. 

We repeat this process again for $r$ and $q$. The result is three double wedges defined by consecutive spokes, one about each site. It follows from our earlier remarks from Section~\ref{sec:hilbert-vor} that within the intersection of these regions, the bisectors are simple conics. We can compute these conics in $O(1)$ times \cite{bumpus2023software} and determine their intersection point, thus yielding the desired point $c$. The radius $\rho$ of the ball can also be computed in $O(1)$ time.

\section{Building the Triangulation} \label{sec:construction}

In this section we present our main result, a randomized incremental algorithm for constructing the Delaunay triangulation $\DT(P)$ for a set of $n$ sites $P$ in the interior of an $m$-sided convex polygon $\Omega$. Our algorithm is loosely based on a well-known randomized incremental algorithm for the Euclidean case~\cite{guibas1992randomized, deberg2010book}. 

\subsection{Orienting and Augmenting the Triangulation} \label{sec:delaun-augment}

In this section we introduce some representational conventions for the sake of our algorithm. First, we will orient the elements of the triangulation. Given $p,q \in \interior(\Omega)$ define the \emph{endpoint} of the $(p,q)$-bisector to be the endpoint that lies to the left of the directed line $\overrightarrow{p q}$. (It follows from the star-shapedness of Voronoi cells that the bisector endpoints intersect $\bd \Omega$ on opposite sides of this line.) The opposite endpoint will be referred to as the $(q,p)$-bisector endpoint. When referring to a triangle $\triangle p q r$, we will assume that the vertices are given in counterclockwise order. Also, edges in the triangulation are assumed to be directed, so we can unambiguously reference the left and right sides of a directed edge $p q$.

As mentioned earlier, the triangulation need not cover the convex hull of the set of sites (see Figure~\ref{fig:delaun-augment}(a)). For the sake of construction, it will be convenient to augment the triangulation with additional elements, so that all of $\Omega$ is covered. First, we add elements to include the endpoints of the Voronoi bisectors on $\bd \Omega$. For each edge $p q$ such that the external face of the triangulation lies to its left, the endpoint of the $(p,q)$-bisector intersects the boundary of $\Omega$. Letting $x$ denote this boundary point, we add a new triangle $\triangle p q x$, called a \emph{tooth}. (See the green shaded triangles in Figure~\ref{fig:delaun-augment}(b).)

\begin{figure}[htbp]
    \centerline{\includegraphics[scale=0.40]{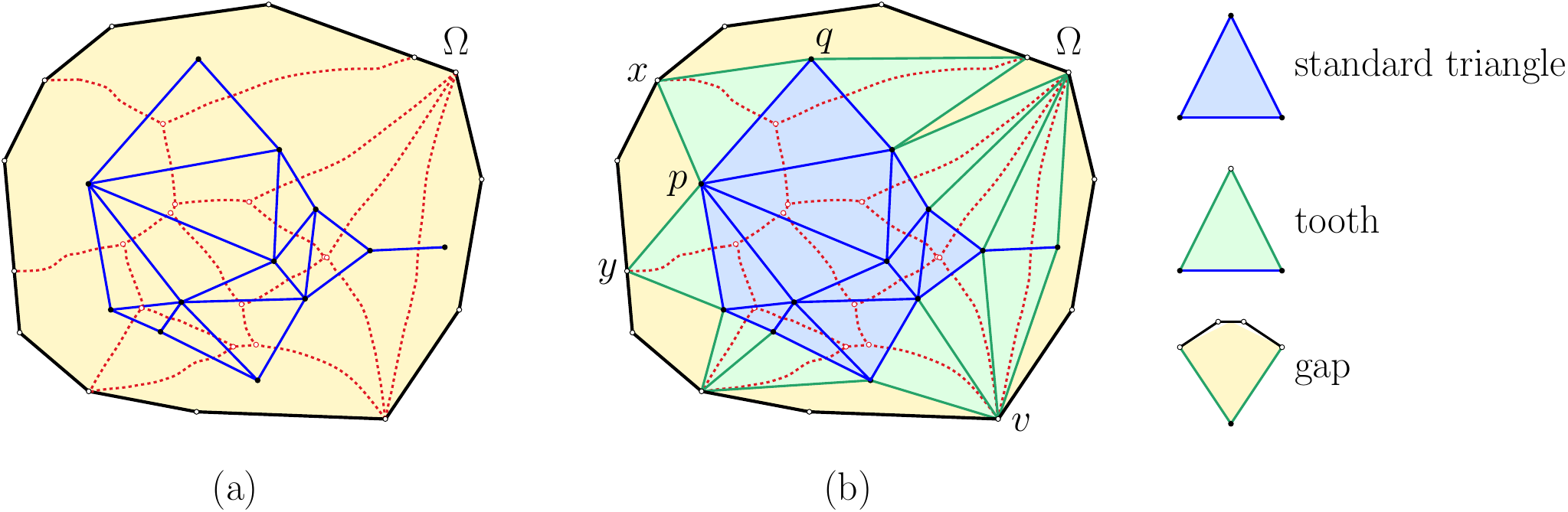}}
    \caption{The augmented triangulation.}
    \label{fig:delaun-augment}
\end{figure}

Finally, the portion of $\Omega$ that lies outside of all the standard triangles and teeth consists of a collection of regions, called \emph{gaps}, each of which involves a single site, the sides of two teeth, and a convex polygonal chain along $\bd \Omega$. While these shapes are not triangles, they are defined by three points, and we will abuse the notation $\triangle p x y$ to refer to the gap defined by the site $p$ and boundary points $x$ and $y$ (see Figure~\ref{fig:delaun-augment}(b)). 

Even when sites are in general position, multiple teeth can share the same boundary vertex. For example, in Figure~\ref{fig:delaun-augment}(b) three teeth meet at the same boundary vertex $v$. In our augmented representation, it will be convenient to treat these as three separate teeth, meeting at three distinct (co-located) vertices, separated by two degenerate (zero-width) gaps. This allows us to conceptualize the region outside of the standard triangulation as consisting of an alternating sequence of teeth and gaps. The following lemma summarizes a few useful facts about the augmented representation.

\begin{restatable}{lemma}{delaunAugment} \label{lem:delaun-augment}
Given a set of $n$ point sites in the interior of a convex polygon $\Omega$:
\begin{itemize}
\item The augmented Delaunay triangulation has complexity $O(n)$.
\item The region covered by standard triangles is connected. 
\item Each standard triangle and each tooth satisfies the empty circumcircle property.
\item The region outside the standard triangles consists of an alternating sequence of teeth and gaps. 
\item For any $p \in \Omega$, membership in any given triangle, tooth, or gap can be determined in $O(1)$ time.
\end{itemize} 
\end{restatable}

\begin{proof} ~ 
\begin{itemize}
\item The standard triangulation is clearly $O(n)$ by standard results on triangulations (based on Euler's formula). The number of teeth is clearly $O(n)$, since each is generated by a standard edge. 
\item The connectivity of the triangulation was proved in Lemma~\ref{lem:planar-spanner}.
\item Through the addition of degenerate gaps, the teeth and gaps alternate. Membership in a standard triangle or tooth can be computed in $O(1)$ time. 
\item Since each gap is bounded by two edges and the boundary of $\Omega$, membership for any $p \in \Omega$ can be determined in $O(1)$ time.
\end{itemize}
\end{proof}

\subsection{Local and Global Delaunay} \label{sec:local-global}

Given a set $P$ of point sites, we say that an augmented triangulation $\TT(P)$ is \emph{globally Delaunay} (or simply \emph{Delaunay}) if for each standard triangle and each tooth $\triangle p q r$, the circumscribing ball, denoted $B(p {\ST} q {\ST} r)$, exists and has no site within its interior. The incremental Euclidean Delaunay triangulation algorithm employs a local condition, which only checks this for neighboring triangles \cite{guibas1992randomized, deberg2010book}. Given a directed edge $p q$ of the triangulation that is incident to two triangles or to a triangle and tooth, let $a$ and $b$ be the vertices of the incident triangles lying to the left and right of the $p q$, respectively. We say that $\TT(P)$ is \emph{locally Delaunay} if $a \notin \interior(B(q {\ST} p {\ST} b))$ and $b \notin \interior(B(p {\ST} q {\ST} a))$. The following lemma shows that it suffices to certify the Delaunay properties locally. The proof follows the same structure as the Euclidean case, but additional care is needed due to the existence of teeth and gaps.

\begin{restatable}{lemma}{localGlobal} \label{lem:local-global}
Given a set $P$ of point sites, an augmented triangulation $\TT(P)$ is globally Delaunay if and only if it is locally Delaunay.
\end{restatable}

\begin{proof}
Our proof follows the same structure as that given in \cite{deberg2010book} for the Euclidean case. Suppose to the contrary that $\TT(P)$ is locally but not globally Delaunay. That is, there exist sites $a, b, c, d \in P$ such that $\TT(P)$ has a triangle or tooth $\triangle a b c$ such that the circumscribing ball $B(a {\ST} b {\ST} c)$ contains $d$ in its interior (see Figure~\ref{fig:local-global}(a)). Let $a b$ be the edge such that $c$ and $d$ lie on opposite sides of $\overline{a b}$. Among all triangle-site pairs $(\triangle a b c, d)$, let this be the one that maximizes $\angle a d b$.

\begin{figure}[htbp]
    \centerline{\includegraphics[scale=0.40]{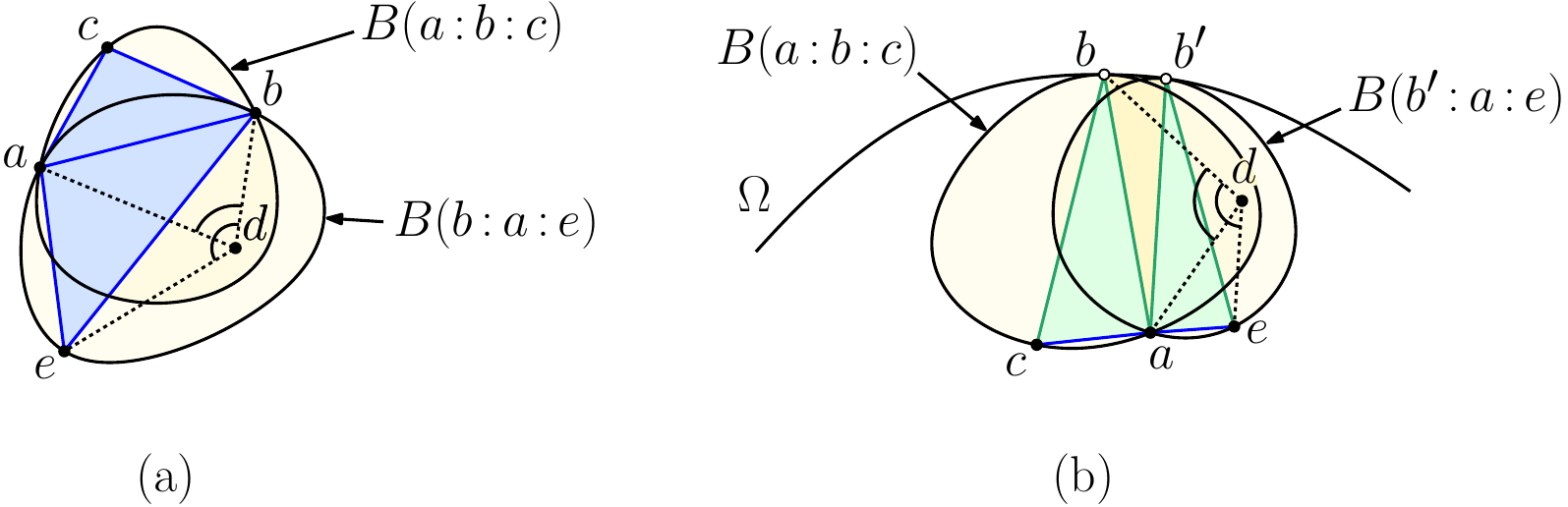}}
    \caption{Proof of Lemma~\ref{lem:local-global}.}
    \label{fig:local-global}
\end{figure}

To complete the proof, we consider two cases. First, if $\triangle a b c$ is a standard triangle or if it is a tooth and both $a$ and $b$ are sites, consider the adjacent triangle $\triangle b a e$ along edge $a b$. Since $\TT(P)$ is locally Delaunay, $e$ does not lie within the interior of $B(a {\ST} b {\ST} c)$. By Lemma~\ref{lem:circum-unique}, the circumscribing ball $B(b {\ST} a {\ST} e)$ completely contains the portion of $B(a {\ST} b {\ST} c)$ that lies on the same side of the chord $\overline{a b}$ as $e$. (If it failed to do so, there would be an additional point of intersection of the boundaries of these balls, other than $a$ and $b$, violating the lemma.) This implies that $d \in B(b {\ST} a {\ST} e)$. Now, let $b e$ denote the edge of $\triangle b a e$ such that $a$ and $d$ lie on the opposite sides of the chord $\overline{b e}$. It follows that $\angle e d b > \angle a d b$, contradicting the definition of the pair $(\triangle a b c, d)$.

In the second case, $\triangle a b c$ is a tooth, and either $a$ or $b$ is a boundary vertex, implying that there is a gap on the other side of $a b$. In this case, relabel the vertices so that $a$ is a site and $b$ is on the boundary. Let $\triangle b' a e$ denote the tooth adjacent to the opposite side of this gap (see Figure~\ref{fig:local-global}(b)). Now, apply the above, using $\triangle b' a e$ in place of $\triangle b a e$. The key property still holds, namely that circumscribing ball $B(b' {\ST} a {\ST} e)$ completely contains the portion of $B(a {\ST} b {\ST} c)$ that lies on the same side of the chord $\overline{a b}$ as $e$.
\end{proof}

\subsection{Incremental Construction} \label{sec:incr-construct}

The algorithm operates by first randomly permuting the sites. It begins by inserting two arbitrary sites $a$ and $b$ and generating the resulting augmented triangulation. In $O(\log^2 m)$ time, we can compute the endpoints of the $(a,b)$-bisector. We add the edge $a b$ and create two teeth by connecting $a$ and $b$ to the bisector endpoints (see Figure~\ref{fig:insert}(a)). We then insert the remaining points one by one, updating the triangulation incrementally after each insertion (see Algorithm~\ref{alg:build-DT}). Later, we will discuss how to determine which element each new site lies in, but for now, let us focus on how the triangulation is updated.

\begin{algorithm}[htbp]
\caption{Constructs the Hilbert Delaunay triangulation of a point set $P$} \label{alg:build-DT}
\begin{algorithmic}
\Procedure{Delaunay}{$P$} \Comment{Build the Delaunay triangulation of point set $P$}
	\State $\TT \gets \text{empty triangulation}$
	\State Randomly permute $P$
	\State $a, b \gets \text{any two points of $P$}$ \Comment{Initialize with sites $a$ and $b$}
	\State $x, y \gets \text{endpoints of the $(a, b)$-bisector}$ \Comment{See Figure~\ref{fig:insert}(a)}
	\State Add edges $a b$, $a x$, $a y$, $b x$, $b y$ to $\TT$
	\ForAll{$p \in P \setminus\{a, b\}$} \Comment{Add all remaining sites}
		\State \Call{Insert}{$p, \TT$}
	\EndFor
	\State \Return $\TT$
\EndProcedure
\end{algorithmic}
\end{algorithm}

\begin{figure}[htbp]
    \centerline{\includegraphics[scale=0.40]{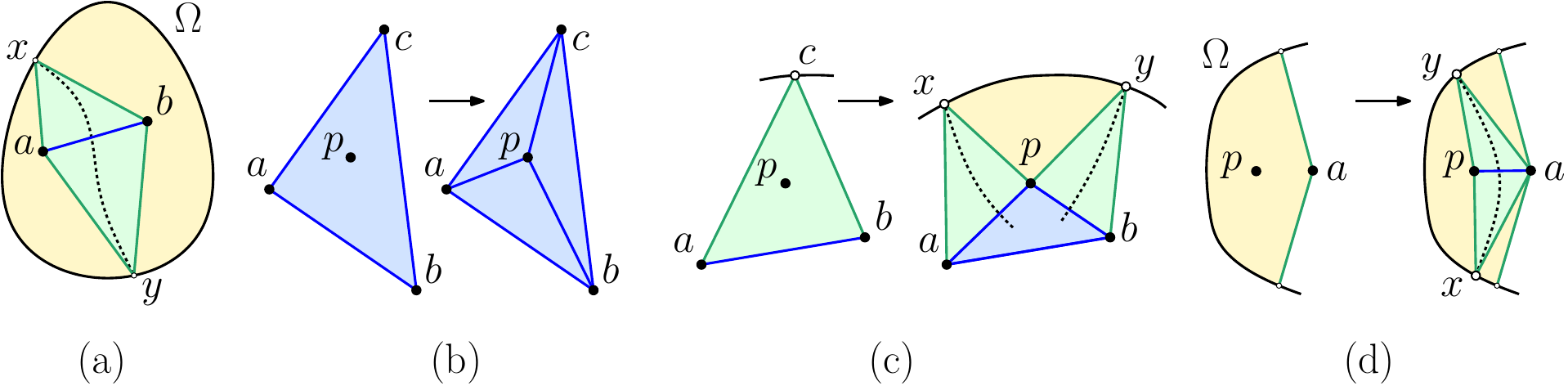}}
    \caption{(a) Initialization and insertion into a (b) standard triangle, (c) tooth, (d) gap.}
    \label{fig:insert}
\end{figure}

When we insert a point, there are three possible cases, depending on the type of element that contains the point, a standard triangle (three sites), a tooth (two sites), or a gap (one site). The case for standard triangles is the same as for Euclidean Delaunay triangulations~\cite{guibas1992randomized}, involving connecting the new site to the triangle's vertices (see Figure~\ref{fig:insert}(b)). Insertion into a tooth involves removing the boundary vertex, connecting the new site to the two existing site vertices, and creating two new teeth based on the bisectors to these sites (see Figure~\ref{fig:insert}(c)). Finally, insertion into a gap involves connecting the new site to the existing site vertex, and creating two new teeth based on the bisectors to this site (see Figure~\ref{fig:insert}(d)).

\begin{algorithm}[htbp]
\caption{Site insertion} \label{alg:insert}
\begin{algorithmic}
\Procedure{Insert}{$p, \TT$} \Comment{Insert a new site $p$ into triangulation $\TT$}
	\State $\triangle a b c \gets \text{the triangle of $\TT$ containing $p$}$
	\If{$\triangle a b c$ is a standard triangle} \Comment{See Figure~\ref{fig:insert}(b)}
    	\State Add edges $a p$, $b p$, $c p$ to $\TT$
            \State \Call{FlipEdge}{$a b, p, \TT$};
            ~\Call{FlipEdge}{$b c, p, \TT$};
            ~\Call{FlipEdge}{$c a, p, \TT$}
	\ElsIf{$\triangle a b c$ is a tooth} \Comment{See Figure~\ref{fig:insert}(c)}
		\State Let $a$ and $b$ be sites, and $c$ be on boundary
            \State $x, y \gets \text{endpoints of the $(a,p)$- and $(p,b)$-bisectors, respectively}$
            \State Remove $c$ and edges $c a$ and $c b$ from $\TT$
            \State Add edges $a p$, $b p$, $a x$, $p x$, $b y$, $p y$ to $\TT$
            \State \Call{FlipEdge}{$a b, p, \TT$};
            \State \Call{FixTooth}{$\triangle a p x, p, \TT$};
            ~\Call{FixTooth}{$\triangle p b y, p, \TT$}
	\Else \Comment{$\triangle a b c$ is a  gap; see Figure~\ref{fig:insert}(d)}
		\State Let $a$ be the site, and let $b$ and $c$ be on the boundary
		\State $x, y \gets \text{endpoints of $(a,p)$- and $(p,a)$-bisectors, respectively}$
            \State Add edges $a p$, $a x$, $p x$, $a y$, $p y$ to $\TT$
            \State \Call{FixTooth}{$\triangle a p x, p, \TT$};
            ~ \Call{FixTooth}{$\triangle p a y, p, \TT$}
    \EndIf 
\EndProcedure
\end{algorithmic}
\end{algorithm}


On return from \textsc{Insert}, $\TT$ is a topologically valid augmented triangulation, but it may fail to satisfy the Delaunay empty circumcircle conditions, and may even fail to be geometrically valid. The procedures \textsc{FlipEdge} and \textsc{FixTooth} repair any potential violations of the local Delaunay conditions. 

The procedure \textsc{FlipEdge} is applied to all newly generated standard triangles $\triangle p a b$. It is given the directed edge $a b$ of the triangle, such that $p$ lies to the left of this edge. It accesses the triangle $\triangle a b c$ lying to the edge's right. This may be a standard triangle or a tooth. In either case, it has an associated circumcircle, which we assume has already been computed.%
\footnote{In the Euclidean case, the in-circle test may be run from either $\triangle a b c$ or $\triangle p a b$, but this is not true for the Hilbert geometry. In light of Lemma~\ref{lem:forbidden}, we do not know whether the $\triangle p a b$ has a circumcircle, so we run the test from $\triangle a b c$.} 
We test whether $p$ encroaches on this circumcircle, and if so, we remove the edge $a b$. If $\triangle{a b c}$ is a standard triangle, then we complete the edge-flip by adding edge $p c$, and then continue by checking the edges $a c$ and $c b$ (see Figure~\ref{fig:flip-and-fix}(a)). Otherwise, we remove vertex $c$, and compute the endpoints $x$ and $y$ of the $(p,a)$- and $(b,p)$-bisectors, respectively. We create two new teeth, $\triangle p a x$ and $\triangle b p y$ (see Figure~\ref{fig:flip-and-fix}(b)). This creates a new (possibly degenerate) gap $\triangle p x y$. Later, we will show (see Lemma~\ref{lem:flip-edge-correct}) that no further updates are needed. The algorithm is presented in Algorithm~\ref{alg:flip-edge}.

\begin{figure}[htbp]
    \centerline{\includegraphics[scale=0.40]{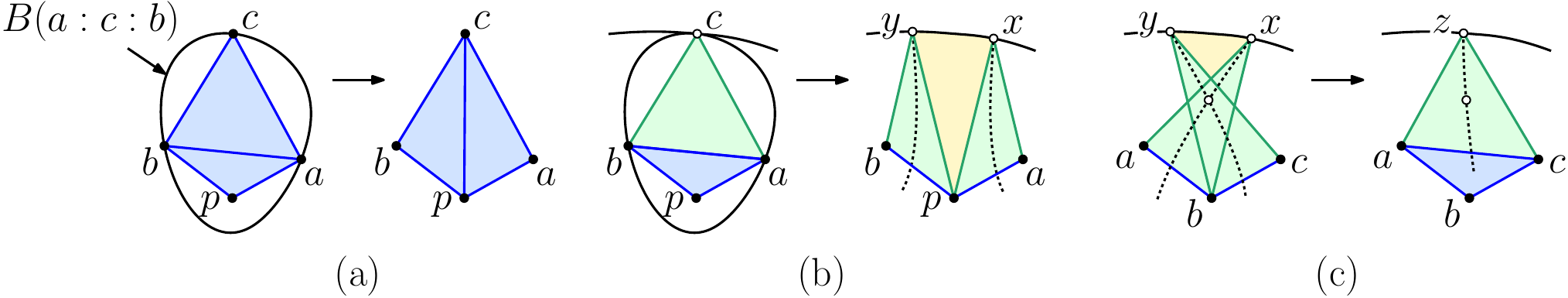}}
    \caption{(a) Standard-triangle edge flip, (b) tooth edge flip, and (c) fixing a tooth.}
    \label{fig:flip-and-fix}
\end{figure}

\begin{algorithm}[htbp]
\caption{In-circle test and edge flip.} \label{alg:flip-edge}
\begin{algorithmic}
\Procedure{FlipEdge}{$a b, p, \TT$} \Comment{In-circle test. New site $p$ lies to the left of edge $a b$.}
        \State $c \gets \text{vertex of triangle to right of $a b$ in $\TT$}$
        \If{$p \in B(a {\ST} c {\ST} b)$} \Comment{$p$ fails the in-circle test for $\triangle a c b$}
            \State Remove edge $a b$ from $\TT$
            \If{$\triangle a c b$ is a standard triangle} \Comment{See Figure~\ref{fig:flip-and-fix}(a)}
                \State Add edge $p c$ to $\TT$
                \State \Call{FlipEdge}{$a c, p, \TT$};
                ~ \Call{FlipEdge}{$c b, p, \TT$} \Comment{Create $\triangle p a c$ and $\triangle b p c$}
             \Else \Comment{$\triangle a c b$ is a tooth. See Figure~\ref{fig:flip-and-fix}(b)}
                \State $x, y \gets \text{endpoints of $(p,a)$- and $(b,p)$-bisectors, respectively}$
                \State Remove $c$ and edges $c b$ and $c a$ from $\TT$
                \State Add edges $a x$, $p x$, $b y$, $p y$ to $\TT$ \Comment{Create teeth $\triangle p a x$ and $\triangle b p y$}
          \EndIf
        \EndIf 
\EndProcedure
\end{algorithmic}
\end{algorithm}


The ``else'' clause creates two teeth $\triangle p a x$ and $\triangle b p y$. The following lemma shows that there is no need to apply \textsc{FixTooth} to them.

\begin{restatable}{lemma}{flipEdgeCorrect} \label{lem:flip-edge-correct}
On return from \textsc{FlipEdge} the teeth $\triangle p a x$ and $\triangle b p y$ generated in the ``else'' clause are both valid.
\end{restatable}

\begin{proof}
Consider the two teeth lying immediately clockwise and counterclockwise from $c$ along the boundary of $\Omega$. Let $x'$ and $y'$ denote the respective boundary vertices of these two teeth. We will show that $x$ lies within the portion of the boundary of $\Omega$ running clockwise from $c$ to $x'$. It will follow from a symmetrical argument that $y$ lies in the counterclockwise portion of the boundary from $c$ to $y'$.

 By the preconditions of \textsc{FlipEdge}, we know that $p$ lies to the left of edge $a b$, and since we have executed the ``else'' clause, we know that $p$ encroaches on the circumscribing ball $B(a {\ST} c {\ST} b)$. Let $\triangle a d x'$ denote the tooth lying clockwise from $c$ prior to the insertion of $p$. If $x$ lies between the $c$ and $x'$, then the circumscribing balls for these triangles can be arranged as shown in Figure~\ref{fig:flip-edge-correct}(a). 

\begin{figure}[htbp]
    \centerline{\includegraphics[scale=0.40]{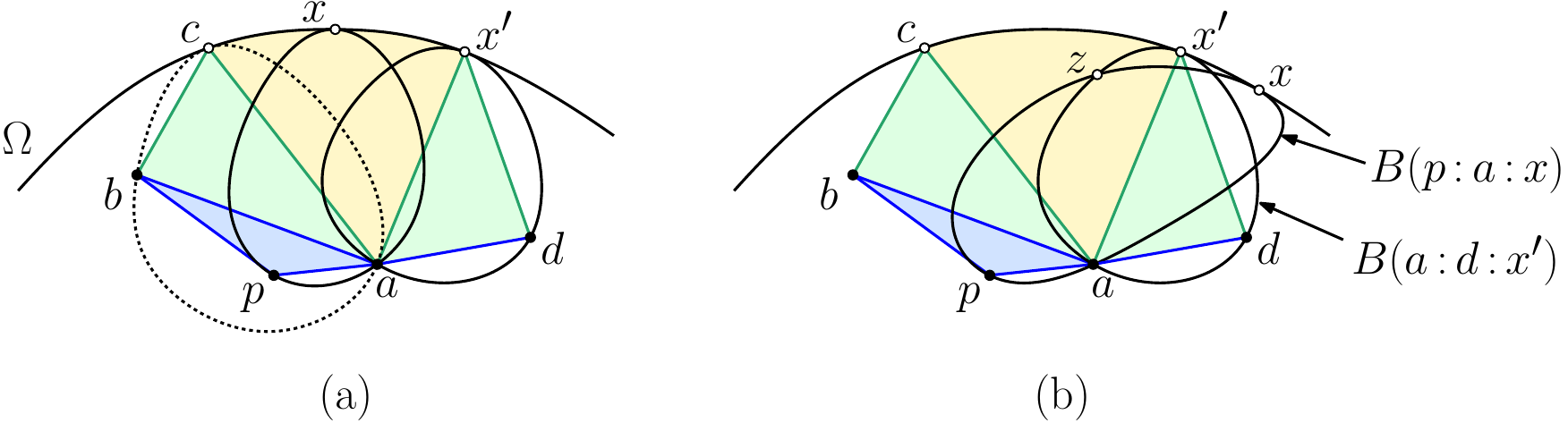}}
    \caption{Proof of Lemma~\ref{lem:flip-edge-correct}.}
    \label{fig:flip-edge-correct}
\end{figure}

Suppose, to the contrary, that $x$ lies on the other side of $x'$ (see Figure~\ref{fig:flip-edge-correct}(b)), observe that the circumcircle $B(p {\ST} a {\ST} x)$ must cross through the circumcircle $B(a {\ST} d {\ST} x')$. It follows that, in addition to $a$ and $d$, these two Hilbert circumcircles intersect in at least one other point $z$. However, this means that there are two different Hilbert circumcircles passing through the three points $a$, $d$, and $z$, which cannot happen according to Lemma~\ref{lem:circum-unique}.
\end{proof}

Finally, we present \textsc{FixTooth}. To understand the issue involved, consider Figure~\ref{fig:insert}(d). In the figure, the newly created teeth $\triangle p a x$ and $\triangle a p y$ both lie entirely inside the existing gap. But, this need not generally be the case, and the newly created boundary vertex may lie outside the current gap, resulting in two or more teeth that overlap each other (see, e.g., Figure~\ref{fig:flip-and-fix}(c)). The procedure \textsc{FixTooth} is given a tooth $\triangle a b x$, where $a$ and $b$ are sites, and $x$ is on $\Omega$'s boundary. Whenever it is invoked the new site $p$ is either $a$ or $b$. Recalling our assumption that teeth and gaps alternate around the boundary of $\Omega$, we check the teeth that are expected to be adjacent to the current tooth on the clockwise and counterclockwise sides. (Only the clockwise case is presented in the algorithm, but the other case is symmetrical, with $a$ and $b$ swapped.) 

Let $\triangle b c y$ denote the tooth immediately clockwise around the boundary. We test whether these triangles overlap (see Figure~\ref{fig:flip-and-fix}(c)). If so, we know that the $(a,b)$-bisector (which ends at $x$) and the $(b,c)$-bisector (which ends at $y$) must intersect. This intersection point is the center of a Hilbert ball circumscribing $\triangle a b c$. We replace the two teeth $\triangle a b x$ and $\triangle b c y$ with the standard triangle $\triangle a b c$ and the tooth $\triangle a c z$, where $z$ is the endpoint of the $(a,c)$-bisector. If $p = a$, then we invoke \textsc{FlipEdge} on the opposite edge $b c$ from $p$, and we invoke \textsc{FixTooth} on the newly created tooth. When the algorithm terminates, all the newly generated elements have been locally validated. The algorithm is presented in Algorithm~\ref{alg:fix-tooth}.

\begin{algorithm}[htbp]
\caption{Check for and fix overlapping teeth.} \label{alg:fix-tooth}
\begin{algorithmic}
\Procedure{FixTooth}{$\triangle a b x, p, \TT$} \Comment{Fix tooth $\triangle a b x$ where $p = a$ or $p = b$}
    \State Let $\triangle b c y$ be the tooth clockwise adjacent to $\triangle a b x$
    \If{$\triangle a b x$ overlaps $\triangle b c y$} \Comment{See Figure~\ref{fig:flip-and-fix}(c)}
        \State $z \gets \text{endpoint of the $(a,c)$-bisector}$
        \State Remove $x$ and $y$ and edges $a x$, $b x$, $b y$, and $c y$ from $\TT$
        \State Add edges $a c$, $a z$, and $c z$ to $\TT$ \Comment{Create triangle $\triangle a b c$ and tooth $\triangle a c z$}
        \If{$p = a$}
            \State \Call{FlipEdge}{$b c, p, \TT$} \Comment{Check edge flip with triangle opposite $b c$}
            \State \Call{FixTooth}{$\triangle a c z, p, \TT$} \Comment{Check for further overlaps}
        \EndIf
    \Else
        \State Repeat the above for the triangle counterclockwise from $\triangle a b x$ swapping $a \leftrightarrow b$
    \EndIf
\EndProcedure
\end{algorithmic}
\end{algorithm}

Based on our earlier remarks, it follows that our algorithm correctly inserts a site into the augmented triangulation.

\begin{restatable}{lemma}{correctness} \label{lem:correctness}
Given an augmented triangulation $\TT(P)$ for a set of sites $P$, the procedure \textsc{Insert} correctly inserts a new site $p$, resulting in the augmented Delaunay triangulation for of $P \cup \{p\}$.
\end{restatable}

The final issue is the algorithm's expected running time. Our analysis follows directly from the analysis of the randomized incremental algorithm for the Euclidean Delaunay triangulation. We can determine which triangle contains each newly inserted point in amortized $O(\log n)$ expected time either by building a history-based point location data structure (as with Guibas, Knuth, and Sharir~\cite{guibas1992randomized}) or by bucketing sites (as with de~Berg {\etal}~\cite{deberg2010book}). The analysis in the Euclidean case is based on a small number of key facts, which apply in our context as well. First, sites are inserted in random order. Second, the conflict set for any triangle or tooth consists of the points lying in the triangle's circumcircle. Both of these clearly hold in the Hilbert setting. Third, the number of structural updates induced by the insertion of site $p$ is proportional to the degree of $p$ following the insertion. This holds for our algorithm, because each modification to the triangulation induced by $p$'s insertion results in a new edge being added to $p$ (and these edges are not deleted until future insertions). Fourth, the structure is invariant to the insertion order. Finally, the triangulation graph is planar, and hence it has constant average degree. The principal difference is that the in-circle test involves computing a Hilbert circumcircle, which by Lemma~\ref{lem:compute-circum} can be performed in $O(\log^3 m)$ time. These additional factors of $O(\log n)$ and $O(\log^3 m)$ are performed a constant number of times in expectation, for each of the $n$ insertions. This implies our main result.

\begin{theorem} \label{lem:delaun-time}
Given a set of $n$ points in the Hilbert geometry defined by a convex $m$-gon $\Omega$, it is possible to construct the augmented Delaunay triangulation in randomized expected time $O(n(\log^3 m + \log n))$.
\end{theorem}

\section{The Hilbert Hull} \label{sec:hilbert-hull}

As observed earlier, the triangles of the Delaunay triangulation of $P$ do not necessarily cover the convex hull of $P$. The region covered by these triangles is called the \emph{Hilbert hull} (the blue region of Figure~\ref{fig:delaun-augment}). In this section, we present a simple algorithm for computing this hull for a set of $n$ points in the Hilbert distance defined by a convex $m$-gon $\Omega$. Our approach is roughly based on the Jarvis march algorithm for computing convex hulls.

Henceforth, when we talk about a site being \emph{on the Hilbert hull}, we mean that it is a vertex of the hull's boundary. These are the sites whose Voronoi cells extend to the boundary of $\Omega$. Our algorithm will output a couterclockwise enumeration of the sites that appear on the hull. By Lemma~\ref{lem:planar-spanner}, the Hilbert hull is simply connected, but its boundary may collapse around a vertex or an edge. If so, a counterclockwise traversal of the boundary may encounter a given site multiple times. (For example, Figure~\ref{fig:hilbert-hull-tree} shows an example of four sites whose Hilbert hull degenerates to a tree. The final hull output would be $\ang{a, b, c, b, d, b, a}$.)

\begin{figure}[htbp]
    \centerline{\includegraphics[scale=0.40]{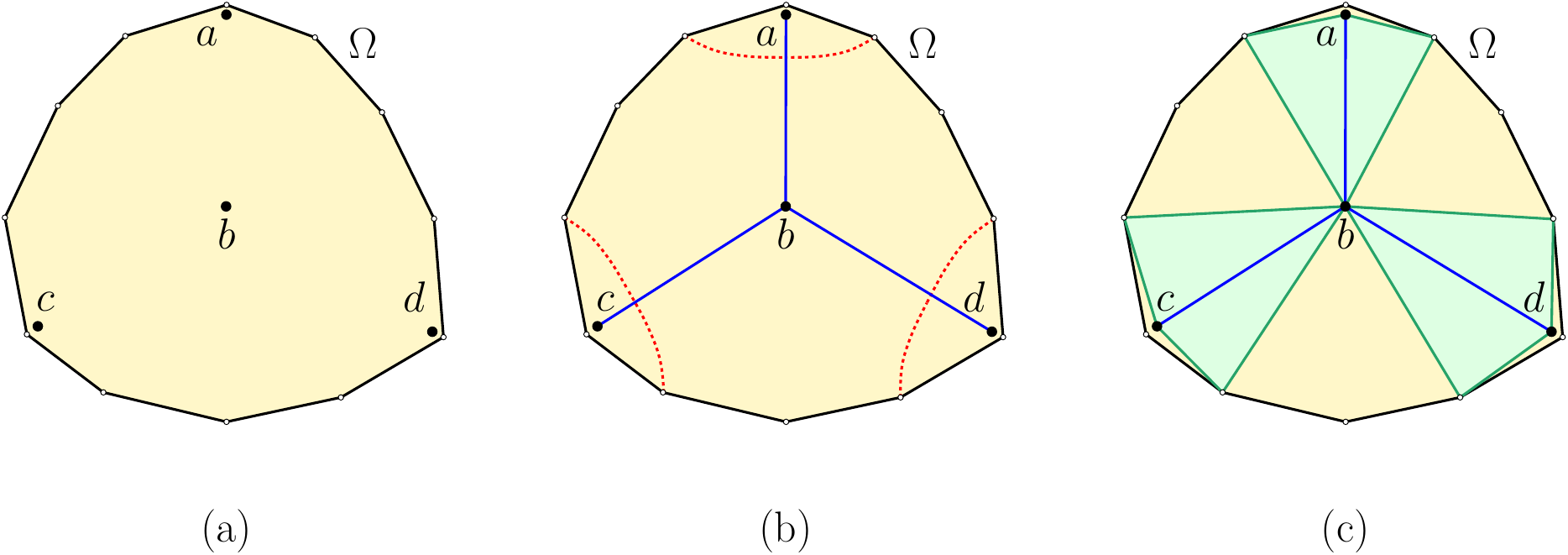}}
    \caption{A set of four sites whose Hilbert hull degenerates to a tree.}
    \label{fig:hilbert-hull-tree}
\end{figure}

Let $p$ be a point on the Hilbert hull of $P$. Since its Voronoi cell extends to the boundary, there exists a point $x \in \bd \Omega$ such that the Hilbert ball at infinity, $B(x,p)$, centered at $x$ and passing through $p$ contains no other points of $P$ (see Figure~\ref{fig:hilbert-hull}(a)). The point $x$ serves as a witness to the fact that $p$ is on the Hilbert hull. 

\begin{figure}[htbp]
    \centerline{\includegraphics[scale=0.40]{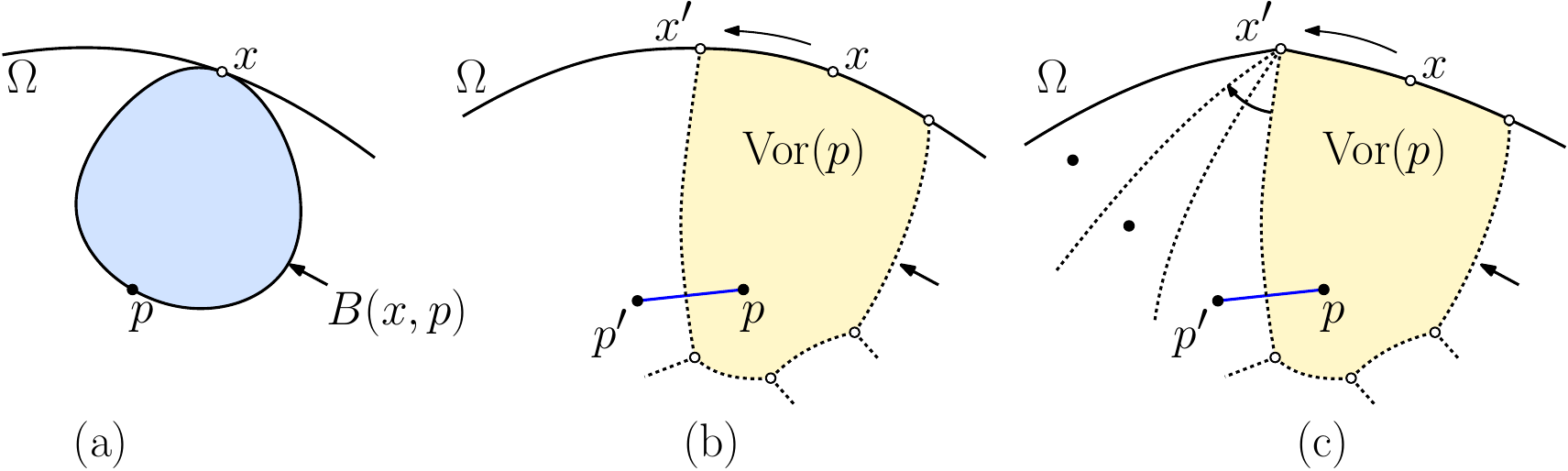}}
    \caption{Computing the Hilbert hull.}
    \label{fig:hilbert-hull}
\end{figure}

The next point of $P$ counterclockwise on the Hilbert hull is the site $p'$ associated with the neighboring Voronoi cell in counterclockwise order around the boundary of $\Omega$ (see Figure~\ref{fig:hilbert-hull}(b)). Let $x'$ denote the endpoint of the $(p',p)$-bisector. In general, multiple Voronoi cells may meet at the same point $x'$, and if so, we can use the angle formed between $x$, $x'$ and $p'$ to break ties (see Figure~\ref{fig:hilbert-hull}(c)).

Based on this observation, our adaptation of Jarvis march involves ordering the sites of $p' \in P \setminus \{p\}$ based on the a lexicographical ordering the bisector endpoints. The points are ordered first by the endpoint of the $(p',p)$-bisector, in counterclockwise order starting at $x$, and ties are broken by considering the clockwise angle $\angle x x' p'$ (ranging over the interval $[0, 2\pi)$). Call this the \emph{ordering induced by $p$ and $x$}.

\begin{lemma} \label{lem:hilbert-hull}
Given a set $P$ of $n \geq 2$ sites in the Hilbert geometry defined by a convex body $\Omega$, a site $p \in P$ that lies on the boundary of $P$'s Hilbert hull, and a point $x \in \bd \Omega$ such that $B(x,p)$ is empty, the next site counterclockwise on the boundary of the Hilbert hull is the site in $P \setminus \{p\}$ that is smallest in the order induced by $p$ and $x$.
\end{lemma}

\begin{proof} 
As observed earlier, the fact that the ball $B(x,p)$ is empty implies that $x$ lies in the (closure) of the Voronoi cell of $p$. Let $x'$ (possibly equal to $x$) be the extreme counterclockwise point of $\Vor(p) \cap \bd \Omega$, and let $p'$ denote the site such that $\Vor(p')$ and $\Vor(p)$ share a common bisector whose endpoint is $x$. The site $p'$ is next on the Hilbert hull, and thus, it suffices to show that $p'$ is the smallest site in the order induced by $p$ and $x$.

For each site $p'' \in P \setminus \{p\}$, let $V(p,p'')$ denote the set of points of $\interior \Omega$ that are closer to $p$ than to $p''$ in terms of Hilbert distance. This is just the set of points that lie on $p$'s side of the $(p'',p)$-bisector. By definition, $\Vor(p) = \bigcap_{p'' \neq p} V(p,p'')$. Clearly then, no site $p'' \neq p'$ can precede $p'$ in the order induced by $p$ and $x$, for otherwise the $(p'',p)$-bisector would intersect the interior of $\Vor(p)$, contradicting the definition of $\Vor(p)$.
\end{proof}

The above lemma can be viewed as a natural way to generalize the ordering used by the standard Jarvis march to the Hilbert context. In the standard setting, consider any point $p$ on the convex hull. The Jarvis ordering is based on the angular order of the directed segments $\overleftrightarrow{p p'}$, for $p' \in P \setminus \{p\}$. If instead, we were to consider the left-directed perpendicular bisectors of each of the pairs $(p',p)$, the ordering would be exactly the same. These left-directed perpendicular bisectors are each directed to a point at infinity, and hence both the Jarvis condition and ours correspond to a sorting of bisector endpoints at infinity.

We can now present our algorithm for computing the Hilbert hull. To start the process off, we identify a pair $(p_0,x_0)$, where $p_0 \in P$, $x_0 \in \bd \Omega$, and $p_0$ lies on the boundary of an empty ball centered at $x_0$. By Lemma~\ref{lem:empty-ball-at-infty}, such a pair can be computed in $O(n \log m)$ time. Assuming that the algorithm has generated an initial sequence points on the hull, arriving at a pair $(p_i,x_i)$ satisfying the above invariant, we compute the next pair as follows. By applying Lemma~\ref{lem:bisector-endpoint}, we compute the bisector endpoints for all the sites of $P \setminus \{p_i\}$. Using the ordering induced by $p_i$ and $x_i$, we take the next point $p_{i+1}$ to be the smallest in this ordering, and let $x_{i+1}$ be the associated bisector endpoint. We repeat this until we get back to (or go beyond) the starting pair $(p_0,x_0)$. The pseudocode is presented in Algorithm~\ref{alg:jarvis}. 

\begin{algorithm}[htbp]
\caption{Hilbert Jarvis march} \label{alg:jarvis}
\begin{algorithmic}
\Procedure{HilbertJarvisMarch}{$P$} \Comment{Compute the Hilbert hull of $P$}
    \State Compute $p_0 \in P$ and $x_0 \in \bd \Omega$ such that $B(p_0,x_0)$ is empty (by Lemma~\ref{lem:empty-ball-at-infty})
    \State $i \gets 0$
    \Repeat
        \State $p_{\text{best}} = \text{null}$
	\For{$p_{\text{test}} \in P \setminus \{p_i\}$} \Comment{Find the next site on the hull}
            \State $x_{\text{test}} \gets \text{endpoint of $(p_{\text{test}},p_i)$ bisector (by Lemma~\ref{lem:bisector-endpoint})}$
            \If{($p_{\text{best}} = \text{null}$) or ($(p_{\text{test}},x_{\text{test}}) < (p_{\text{best}},x_{\text{best}})$ in the $(p_i,x_i)$ ordering)}
                \State $(p_{\text{best}},x_{\text{best}}) \gets (p_{\text{test}},x_{\text{test}})$
            \EndIf
        \EndFor
        \State $(p_{i+1},x_{i+1}) \gets (p_{\text{best}},x_{\text{best}})$
        \State $i \gets i+1$
    \Until{$(p_i,x_i)$ goes beyond $(p_0,x_0)$} \Comment{We've returned to the starting point}
    \State \Return $\ang{p_0, p_1, \ldots p_i}$ \Comment{Return the final hull sequence}
\EndProcedure
\end{algorithmic}
\end{algorithm}

The correctness of this algorithm follows directly from Lemma~\ref{lem:hilbert-hull}. With each iteration, the algorithm invokes Lemma~\ref{lem:bisector-endpoint} on each of the $n-1$ sites other than $p_i$, for a running time of $O(n \log^2 m)$. Letting $h$ denote the number of edges on the Hilbert hull, the algorithm terminates after $h$ iterations. This yields the following result.

\begin{lemma} \label{lem:hull-time}
Given a set of $n$ points in the Hilbert geometry defined by a convex $m$-gon $\Omega$, it is possible to construct the Hilbert hull in time $O(n h \log^2 m)$, where $h$ is the number of edges on the hull's boundary.
\end{lemma}

\bibliography{shortcuts,hilbert}

\end{document}